\documentclass[journal,comsoc]{IEEEtran}

\usepackage[utf8]{inputenc} 
\usepackage[T1]{fontenc}
\usepackage{url,comment}
\usepackage{ifthen}
\usepackage{cite}
\usepackage{graphicx}
\usepackage{amsmath,amssymb,amsthm}
\usepackage{cite,soul}

\usepackage{multirow,rotating}
\usepackage{diagbox}
\usepackage{threeparttable}
\usepackage{booktabs}

\DeclareMathOperator{\sexpD}{S-Exp}

\newtheorem{theorem}{Theorem}
\newtheorem{lemma}{Lemma}
\newtheorem{corollary}{Corollary}

\usepackage[colorlinks,pdfstartview=FitH,citecolor=blue,linkcolor=blue,urlcolor=blue]{hyperref}
\usepackage{tikz}
\usepackage[capitalize]{cleveref}

\usepackage[normalem]{ulem}
\usepackage{xcolor}
\makeatletter
\def\squiggly{\bgroup \markoverwith{\textcolor[rgb]{1,0,0}{\lower3.5\p@\hbox{\sixly \char58}}}\ULon}
\makeatother

\newcommand{\prob}[1]{{\mathbb P}}
\ifCLASSINFOpdf
\else
\fi

\begin{document}
%
\title{Covert, Low-Delay, Coded Message Passing\\ in Mobile (IoT) Networks
}
%
%
%

\author{Pei Peng, Emina Soljanin,~\IEEEmembership{Fellow,~IEEE}
\thanks{P.~Peng and E.~Soljanin are with the Department
of Electrical and Computer Engineering, Rutgers, The State University of New Jersey, Piscataway, NJ 08854, USA, e-mail: \{pei.peng, emina.soljanin\}@rutgers.edu.

Some parts of Sec.~\ref{Sec:analysis_cp} and \ref{Sec:rw} of this paper appeared in the 2019 57th Annual Allerton Conf.\ on Communication, Control, and Computing \cite{peng2019straggling}.

Part  of  this  research  is  based  upon  work  supported  by the National Science Foundation under Grant No.~SaTC-1816404.
}
}

%
%

\markboth{submitted to IEEE Transactions on Information Forensics and Security, August~2021} {}
%



\maketitle

\begin{abstract}
We introduce a  gossip-like protocol for covert message passing between Alice and Bob as they move in an area watched over by a warden Willie. The area hosts a multitude of Internet of (Battlefield) Things  (Io$\beta$T) objects. Alice and Bob perform random walks on a random regular graph. The Io$\beta$T objects reside on the vertices of this graph, and some can serve as relays between Alice and Bob. The protocol starts with Alice splitting her message into small chunks, which she can covertly deposit to the relays she encounters. The protocol ends with Bob collecting the chunks. Alice may encode her data before the dissemination. Willie can either perform random walks as Alice and Bob do or conduct uniform surveillance of the area. In either case, he can only observe one relay at a time. We evaluate the system performance by the covertness probability and the message passing delay.  In our protocol, Alice splits her message to increase the covertness probability and adds (coded) redundancy to reduce the transmission delay. The performance metrics depend on the graph, communications delay, and code parameters. We show that, in most scenarios, it is impossible to find the design parameters that simultaneously maximize the covertness probability and minimize the message delay. 
\end{abstract}

\begin{IEEEkeywords}
Covert communications, Random walks, Internet of (battlefield) things, delay reduction. 
\end{IEEEkeywords}


%
\IEEEpeerreviewmaketitle

\section{Introduction}
\label{Sec:intro}

Hiding various aspects of communications is often essential. In wartime, communication between the suspected parties can alert the adversary even if the message is unknown.  In everyday life, revealing the identity of communicating parties affects the increasingly important anonymity and privacy.  Several recent papers addressed covert communications at the physical layer. There, two parties, Alice and Bob, communicate while observed by the warden Willie. An information-theoretic approach to achieving covertness, roughly speaking, relies on camouflaging messages as noise (see, e.g., \cite{bash2013limits, HidingInformationInNoise:BashGT15, CovertComm:Bloch16, DeniableComm:KadheJB14,wang2016fundamental} and references therein).  An extension of this model (see, e.g., \cite{sobers2017covert}) involves a jammer that can help Alice to transmit covertly \cite{shahzad2017covert}. Another extension involves a third participant Carol which acts as a cover for Bob, \cite{hu2018covert,hu2019covert} studies the covert communication in one-way relay networks. 
We here propose a covert message-passing protocol for IoT environments. This protocol is complementary and can be used in conjunction with the previously proposed methods for covert transmission.

The last decade has seen a wide variety of novel communications systems. Future 5G systems are supposed to host a hundred times more devices than current 4G environments, and one can potentially harness the resources expected to be brought in by smart (battlefield) devices in the future Internet of (Battlefield) Things (Io$\beta$T) environments. By exploiting these devices' storage and communication ability, Alice can covertly pass messages to Bob in a gossip-like manner as outlined below. 

Alice and Bob communicate over a wide geographic area (e.g., battlefield or an occupied city) patrolled by a warden Willie. The area hosts a multitude of IoT objects capable of storing, sending, and receiving data. Alice splits her message into small chunks, which she can inconspicuously pass, one at a time, to IoT objects that appear in her proximity as she randomly moves through the area. We refer to such IoT objects as relays or helpers. Bob, who also randomly moves through the area, can then retrieve the stored data chunks. Because the IoT objects are distributed over a wide area, Willie can only periodically check if any of these objects is transmitting or receiving data.

The movements of Alice and Bob are modeled as random walks on graphs. Information gathering and dissemination on graphs is an interesting problem that naturally arises and is actively researched in many different contexts. Examples include: border control using unmanned aerial vehicles (UAVs) \cite{girard2004border}, measuring traffic, reporting road conditions and helping with emergency response using UAVs \cite{puri2005survey}, monitoring the ocean \cite{paley2008cooperative}, measuring air pollution \cite{villa2016development}, multi-agent systems \cite{soljanin2010reducing}, and more recently for timely exchange of information updates \cite{tripathi2021age}.
We consider two detection patterns of Willie: random patrolling and uniform surveillance. In the patrolling model, Willie performs random walks on graphs and can detect the communication when he happens to be on at the same relay as Alice or Bob. In the surveillance model, Willie scans graph nodes (e.g., while sitting in a control room). He can detect communication when he observes Alice or Bob.

Having to distribute and collect many chunks, as well as 
the unpredictability of mobility and availability of helpers can cause significant delays in our mobile information transfer. To increase the persistence of information in the unstable environment, the agent may make the data chunks redundant by erasure correcting codes, requiring more data chunks to be distributed and collected. One would expect that would further increase the delay. However, that is not necessarily the case, and we will see that coding and some other forms of redundancy can, in fact, be used to reduce the delay, as previously shown to be the case in data download and straggler mitigation (see, e.g., \cite{joshi2012coding,scale:MehmetToN19,peng2020diversity,yu2020straggler,badita2020optimal,replication:AmirToN21} and references therein).

It may be helpful to the reader to consider our work in the context of the literature on throughput and delay in mobile ad-hoc communications.
For example, \cite{gupta2000capacity} proposed a model for studying the throughput scaling of static ad hoc networks; \cite{grossglauser2002mobility} introduced the mobility into the communications model and assumed the source only transfers the data packet to a nearby relay which subsequently passes the packet to the destination.
\cite{el2006optimal} introduced a random walk model for node mobility and analyzed the throughput and delay tradeoff, and \cite{codingTD:KongYS12} showed how coding improves this tradeoff, and \cite{im2020mobility} studied the throughput scaling of covert communication in ad-hoc networks  introduced in \cite{grossglauser2002mobility}. Our work is different in multiple ways. It considers an  Io$\beta$T specific system model, which includes different communication protocols and performance metrics. However, our approach to deriving communications delay and the techniques for reducing it should be of interest to the mobile ad-hoc communications community.

The contributions of this paper are summarized as follows:
\begin{enumerate}
\item We first propose a gossip-like protocol for a covert dissemination/collection of message chunks in IoT environments. We analyze the dissemination and collection delay for two chunk transmission time models. For both models, our theoretical and numerical results show that introducing redundant IoT relays can reduce the dissemination time, and erasure coding of messages can reduce the collection time.

\item We then extend our analysis to covert communication scenarios. We introduce two warden models and derive/analyze the covertness probability for each. The theoretical and numerical results show that coding always reduces the covertness probability while splitting the message into smaller chunks may increase the covertness probability. 

\item We numerically analyze delay vs.\ the covertness probability tradeoff. We conclude that the tradeoff is very different for different system parameters. In some scenarios, there exists an optimal code rate that maximizes the covertness probability and minimizes delay. However, in most scenarios, simultaneously maximizing the covertness probability and minimizing the communications delay is impossible.
\end{enumerate}

The paper is organized as follows: In Sec.~\ref{Sec:sys}, we present the message passing model and two delay models. In Sec.~\ref{Sec:messagedelay}, we derive and analyze the expressions of dissemination/collection time for different delay models and numerically analyze the message passing delays.
In Sec.~\ref{Sec:covertpass}, we propose two warden detection models and point out the tradeoff that exists between the message passing delay and covertness probability. 
In Sec.~\ref{Sec:analysis_cp}, we derive and analyze the expressions of covertness probability for different detection models.
In Sec.~\ref{Sec:tradeoff}, we present some numerical results for the delay vs. covertness probability tradeoff. Conclusions are given in Sec.~\ref{Sec:conclusion}.

\section{System Model and Problem Formulation}
\label{Sec:sys}

We first describe the communication participants, their mobility, and  
message passing protocols and delay models. We then define two message passing performance metrics we will be studying. This section focuses the message passing delay. The message passing covertness and the delay/covertness tradeoff are discussed in Sec.~\ref{Sec:covertpass}. 

\subsection{Communication Participants and their Mobility}
There are three types of participants present within a geographical area (e.g., a city or a battlefield): a mobile source Alice, $r$ static relays, a mobile receiver Bob. Alice's and Bob's mobility is modeled as a simple random walk over a {\it mobility graph}. (Sec.~\ref{Sec:covertpass} extends this model to include a warden Willie.)

We model the mobility graph as a random $d$-regular graph, that is, a graph selected uniformly at random from the set 
of all $d$-regular graphs on $n$ vertices. 
The related literature (see, e.g., \cite{el2006optimal, codingTD:KongYS12}) uses rectangular grid graphs over a torus to study communications in (mobile) ad-hoc networks. We adopt the random regular graph model for the following reasons. 1) Since these graphs are locally trees with high probability \cite{cooper2009multiple}, the random walks on such graphs are reasonable mobility models. 2) Although random regular graphs are more realistic than the grid on torus graph for mobility modeling, random walks on these graphs are fairly well understood. Thus, the existing mathematical results on this topic can be used and easily extended to our scenarios, allowing us to concentrate on the communications and covertness problems we aim to solve. Other mobility graphs will be used in future research, especially irregular graphs with a few high degree hubs.

The $r$ relays are IoT objects residing on $r$ vertices of the mobility graph.  Alice has a message to pass to Bob. Instead of communicating with Bob directly, Alice uses the IoT objects (relays) to store pieces of her message, which Bob can subsequently retrieve. Because IoT objects have limited power, Alice and Bob can only communicate with a relay when they reach the vertex where the relay resides. 


\subsection{Communications Protocol \& the Mobility Model}

Alice's \textit{message} has length $m$ bits (symbols, packets). Because of the relay's storage constraints and to avoid long transmission time at a relay, she partitions the message into $k$ {\em data chunks} each of length $\ell=m/k$. By using an MDS code, she encodes the data chunks into $n(\le r)$ {\em coded chunks} (also of length $\ell$). Thus, Bob needs to collect any $k$ out of $n$ coded chunks to recover the message. An example of message partitioning and encoding, together with the mobility graph, is shown in Fig.~\ref{fig:sysmodel}. Here, Alice splits the message into $k=2$ data chunks and encodes them into $n=3$ coded chunks $\{M_1,M_2,M_3\}$.
\begin{figure}[hbt]
    \centering
    \includegraphics[width=0.48\textwidth]{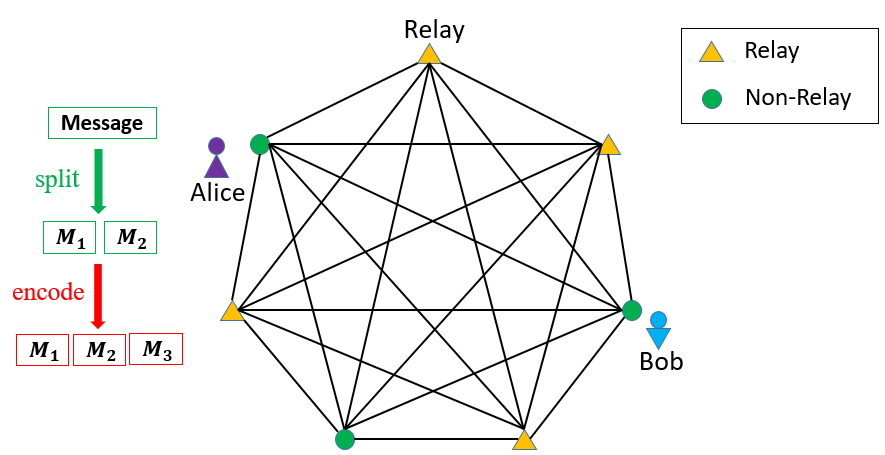}
    \caption{An example of message passing over a complete graph with with $v=7$ vertices. Each vertex contains an IoT object. Four among them (yellow triangles) are relays and three (green disks) are not. Alice and Bob can communicate with any object but only disseminate/collect message chunks to/from the relays. }
    \label{fig:sysmodel}
\end{figure}

Message passing from Alice to Bob has two phases: the dissemination phase and the collection phase. In the dissemination phase, Alice transfers $n$ coded chunks to the first $n$ relays she encounters as she randomly walks through the graph. This model is similar to a classic predator-prey model, e.g., as analyzed in  \cite{cooper2009multiple}. Fig.~\ref{fig:dismodel} shows an example. When she meets the first relay, she deposits the coded chunk $M_1$ to the relay. She then deposits $M_2$ to the second relay, and so on. The dissemination is complete when Alice deposits the last chunk. Fig.~\ref{fig:dismodel} shows the snapshots of Alice at the three relays. She may need to walk through several non-relay vertices to reach all the relays.
 \begin{figure}[hbt]
    \centering
    \includegraphics[width=0.48\textwidth]{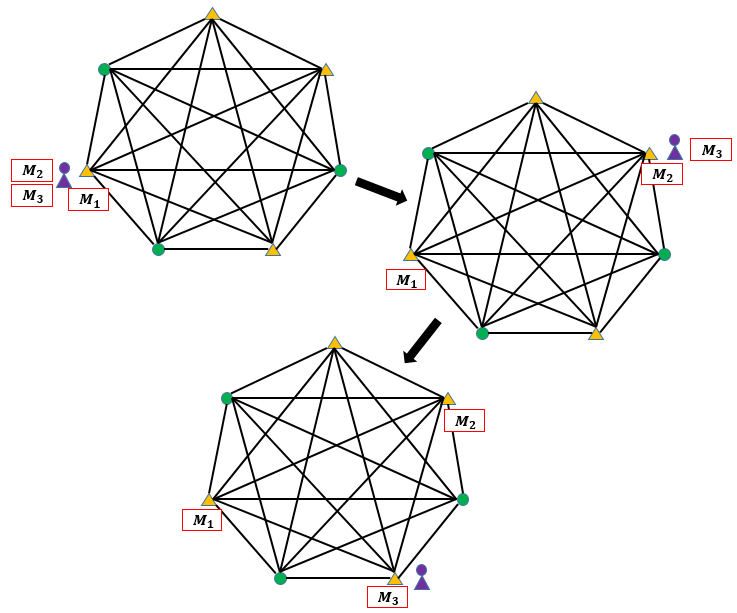}
    \caption{Alice needs to disseminate three coded chunks to any three out of four relays. Every time when she meets a relay, she stores one coded chunk, until all three coded chunks are disseminated. Each relay can only store one coded chunk.}
    \label{fig:dismodel}
\end{figure}
Similarly, Bob recovers the message in the collection phase by collecting $k$ coded chunks from the $n$ relays as he randomly walks through the graph. 

\subsection{Message Chunk Passing Time}
\label{subsec:chucktrans}

The time to pass a chunk between Alice (or Bob) and a relay has two components. The first is the time Alice (or Bob) needs to spend randomly walking to meet a relay. The second is the transmission time of a data chunk between Alice/Bob and a relay. We refer to Alice's move from a vertex to one of its neighbors as {\it one step}. The time to take a step is, in general, a random variable $\eta$.
We consider the following two chunk passing time models:  

\noindent \ul{Model 1} -- {\it Constant Transmission Time:} \space
 In this classic model for random walks (see, e.g., \cite{cooper2009multiple,alon2008many} and references therein), the chunk passing time is measured in the number of steps Alice/Bob needs to make to encounter a relay. It is applicable here under the following circumstances: 1) the time of a step $\eta$ is a constant, say $1$ and 2) the chunks are sufficiently small so that their transmission time can be neglected. 

\noindent \ul{Model 2} -- {\it Random Transmission Time:} \space
In this general model, the step times are independent and identically distributed (i.i.d.) random variables. The transmission time of a data chunk between Alice/Bob and a relay follows a shifted exponential distribution whose shift is proportional to the length of a chunk $\ell$ and whose tail accounts for various disturbances (noise) in the system.

\subsection{Covert Message Passing}
\label{subsec:syscovert}
The covert message passing is defined as Alice successfully passing a message to Bob through some relays without Willie's detection. We consider two detection patterns of Willie: random patrolling and uniform surveillance. In the patrolling model, Willie performs random walks on graphs and can detect the communication when he happens to be on at the same relay as Alice or Bob. In the surveillance model, Willie scans graph nodes (e.g., while sitting in a control room). He can detect communication when he observes Alice or Bob. The details of the detection models are given in Sec.~\ref{Sec:covertpass}.

\subsection{Performance Metrics and Problem Formulation}
\label{sec:metrics}
Alice's goal is to covertly and quickly pass the message to Bob. Thus, 
the performance metrics of interest are \emph{covertness probability} ($P_c$), the \emph{expected dissemination time} ($\mathbb{E}[T_A]$) of $n$ coded chunks by Alice, the \emph{expected collection time} ($\mathbb{E}[T_B]$) of $k$ coded chunks by Bob, and the  \emph{expected message passing delay} ($\mathbb{E}[T_{A+B}]$) (dissemination plus collection). We evaluate the above metrics in terms of the design parameters $k$ and $n$, where $k/n$ is the code rate. Our goal is to find the optimal $k$ and $n$ that maximize $P_c$, and minimize $\mathbb{E}[T_A]$, $\mathbb{E}[T_B]$, and $\mathbb{E}[T_{A+B}]$.
In the following section, we will derive the formulas for each metric and analyze the optimal $k$ and $n$. The covertness probability will be addressed in Sec.~\ref{Sec:covertpass}.

Recall that a {\it message} of length $m$ data bits, and is split into $k$ {\it data chunks}. Applying an $[n,k]$ MDS code, $k$ data chunks are encoded into $n$ {\it coded chunks}. The length of each data/coded chunk is $\ell=m/k$ bits. Therefore, Alice needs to meet $n$ relays to deposit her coded chunks. We denote by $S_A$ the random number of steps Alice needs to make to meet $n$ out of $r$ relays. Bob needs to meet $k$ relays with coded chunks. We denote by $S_B$ the random number of steps Bob needs to make to meet $k$ out of $n$ relays storing coded chunks.

\subsection{Parameters and Notation}
\begin{center}
\begin{small}
    \begin{tabular}{rcl}
       $v$ & -  &  number vertices in the mobility graph\\
       $d$ & - & mobility graph vertex degree\\
       $r$ & -  & number of relays in the system\\
       $m$  & -  & message (data) length in bits\\
       $k$  & -  & number of message (data) chunks \\
    $n$  & -  & number of coded chunks \\
    $k/n$  & -  & code rate \\
    $\ell$  & -  & length of a data (or coded) chunk\\ 
    $\eta$  & -  & random walking time between two vertices\\ 
    $T_A$ & -  & random dissemination time of $n$ coded chunks\\
    $T_B$ & -  & random collection time of $k$ coded chunks\\
    $T_{A+B}$ & -  & random message passing time\\
    $S_A$ & -  & random \# of Alice's steps to meet $n$ relays\\
    $S_B$ & -  & random \# of  Bob's steps to meet $k$ relays
    \end{tabular}
    \end{small}
\end{center}

Observe that $v$, $d$, $r$, $m$ and $\mathbb{E}[\eta]$ are given system parameters. $\mathbb{E}[T_A]$, $\mathbb{E}[T_B]$, $\mathbb{E}[T_{A+B}]$, $\mathbb{E}[S_A]$ and $\mathbb{E}[S_B]$ are the performance metrics. $n$, $k$ and $\ell$ are the design parameters which we can select to optimize the performances metrics.

\section{Message Passing Delays}
\label{Sec:messagedelay}
\subsection{Preliminary Reasoning}
Message passing delay includes dissemination and the collection time. It is the the time during which Alice and Bob communications with relays could be discovered by Willie. Observe that the message passing delay is not the time that the message spends in the system from the beginning of Alice's dissemination to the end of Bob's collection. That time depends on the starting times of the dissemination and collection, and is beyond the scope of this paper.

The message passing delay depends on how many steps Alice (Bob) needs to make to disseminate (collect) the coded chunks.To understand the message passing delay dependencies on the system and design parameters, we first consider a special case when the mobility graph is a complete graph. 
Here, Alice needs to disseminate $n$ coded chunks to $r$ relays residing on a $v$-vertex complete graph. Her first chunk can be deposited on any of the $r$ relays. In a single step, Alice arrives at a relay with the probability $r/v$, and thus she needs to make on average $v/r$ steps to find a relay to deposit her first chunk. For depositing her second chunk, Alice needs to arrive to one of the remaining $r-1$ relays, which happens in a single step with probability $(r-1)/v$. Thus, she has to make another $v/(r-1)$ steps on average to deposit her second chunk. Therefore, to deposit her first two chunks, Alice will make  $\frac{v}{r}+\frac{v}{r-1}$ steps on average. Following this reasoning, we see that for $n$ chunks, Alice needs to makes on average $v(H_r-H_{r-n})$ steps, where $H_{r}=\sum_{i=1}^{r}1/i$ is the $r$-th harmonic number. We often use the approximation $H_n=\log n+\gamma + \mathcal{O}(n^{-1})$, where
$\gamma = 0.577$ is Euler's constant.  Bob needs to collect $k$ chunks from $n$ relays storing coded chunks. Therefore, Bob needs to make on average $v(H_n-H_{n-k})$ steps.

The above reasoning does not extend to general graphs. However, reference \cite{cooper2009multiple} provides useful results for large, random regular graphs (which are good mobility graphs models). We adapt the findings of \cite{cooper2009multiple} to our setting, and get the following results. Let $S$ be the number of steps that the source (or collector) needs to make to meet a relay. When we randomly choose a graph from the set of all $d$-regular graphs with $v$ vertices, with high probability, 1) the expected number of steps is 
\begin{equation}
\label{Eq:slot}
\mathbb{E}[S]\sim_{v}\frac{\theta_d v}{r}
\end{equation}
and 2) the probability that the source meets a relay $P_{meet}$ is
\begin{equation}
\label{Eq:meet}
P_{meet}\sim_{v}\frac{r}{\theta_d v}.
\end{equation}
where $X\sim_{Z}Y $ means $\lim_{Z\rightarrow \infty}X/Y=1$, and $\theta_{d}=\frac{d-1}{d-2}$.

We can now find the the expected message passing delay. According to the definition of $T_{A+B}$, we know that its expectation is $\mathbb{E}[T_{A+B}]=\mathbb{E}[T_A] + \mathbb{E}[T_B]$. The dissemination time $T_A$ is the sum of $n$ message chunk passing times (discussed in Sec.~\ref{subsec:chucktrans}). For the constant transmission time model with the step time $\eta=1$, $T_A$ is equal to the random number of steps $S_{A}$ that Alice needs to make to meet $n$ out of $r$ relays to deposit her $n$ coded chunks. Therefore, the expected dissemination time is $\mathbb{E}[T_{A}]=\mathbb{E}[S_{A}]$. Similarly, the collection time $T_B$ is equivalent to the total number of steps $S_{B}$ that Bob needs to make to meet $k$ out of $n$ relays to collect $k$ coded chunks. Therefore, the expected dissemination time is $\mathbb{E}[T_{B}]=\mathbb{E}[S_{B}]$.

For the random transmission time model, the dissemination time $T_A$ (collection time $T_B$) is again equal to the sum of $n$ ($k$) chunk passing times. Recall that the chunk transmission time can not be neglected in this model, and it follows a shifted exponential distribution. Thus the calculations of $\mathbb{E}[T_A]$ and $\mathbb{E}[T_B]$ are more complicated and will be discussed later. 

In the following, we focus on the delay analysis of a random regular graph. We will find the expression of message passing delay by deriving the dissemination time and collection time separately under two different delay models.

\subsection{Constant Transmission Time}
\label{Sec:random}
Under the constant transmission time model, the message passing delay $T_{A+B}=S_{A}+S_{B}$, where $S_{A}$ and $S_{B}$ are the number of steps that the source and collector need to spend to meet enough relays. 

\subsubsection{Dissemination Time}
From the above discussion, it follows that in a $d$-regular graph with $v$ vertices, the source needs to make on average $\frac{\theta_d v}{r}$ steps to disseminate the first coded chunk. After the first dissemination, the number of available relays reduces to $r-1$, and then the source needs to spend on average $\frac{\theta_d v}{r-1}$ steps to disseminate the second chunk. 
Finally, we can get the expected number of steps to disseminate $n$ coded chunks
\[
\mathbb{E}[S_{A}]\sim_{v}\theta_d v (H_r-H_{r-n}).
\]
It is obvious that $\mathbb{E}[S_{A}]$ decreases with increasing $r$ and increases with $n$. Since $n=\varphi r$ for some $\varphi\in(0,1)$, it is not hard to see $\mathcal{O}(\mathbb{E}[S_{A}])=\mathcal{O}(v)$.

For a scenario without redundant relays, the number of relays $r$ is equal to the number of coded chunks $n$. Thus, the expected number of steps to disseminate $n$ coded chunks is $\mathbb{E}[S_{A}]\sim_{v}\theta_d v H_n$. Namely, we have  $\mathcal{O}(\mathbb{E}[S_{A}])=\mathcal{O}(v\log n)$.

\subsubsection{Collection Time}
The collector needs to collect any $k$ coded chunks from $n$ relays that store the chunks. Similarly, we can get the expected number of steps to collect any $k$ chunks
\[
\mathbb{E}[S_{B}]\sim_{v}\theta_d v (H_n-H_{n-k}).
\]
It is obvious that $\mathbb{E}[S_{B}]$ decreases with increasing $n$ and increases with $k$. Since $k=\xi n$ for some $\xi\in(0,1)$, it is not hard to see $\mathcal{O}(\mathbb{E}[S_{B}])=\mathcal{O}(v)$.

For a scenario without coding, the code rate $k/n=1$, namely, $k=n$. Thus, the expected number of steps to collect $n$ data chunks is $\mathbb{E}[S_{B}]\sim_{v}\theta_d v H_n$. Namely, we have  $\mathcal{O}(\mathbb{E}[S_{B}])=\mathcal{O}(v\log n)$.

\subsubsection{Message Passing Delay}
The message passing delay is the sum of dissemination time and collection time. Since both $\mathbb{E}[S_{A}]$ and $\mathbb{E}[S_{B}]$ are positive, the following asymptotic equivalence holds:
\begin{equation}
\label{Eq:small}
\mathbb{E}[T_{A+B}]\sim_{v}\theta_d v (H_r+H_n-H_{r-n}-H_{n-k}).
\end{equation}
We know that $\mathbb{E}[T_{A+B}]$ increases with $k$. Namely, when the number of coded chunks $n$ is given, $\mathbb{E}[T_{A+B}]$ reaches minimum at $k=1$. However, since the IoT devices have limited storage, we sometimes have to split the original message into smaller chunks to store in the relays. Therefore, it is important to study when $k$ is given, what is the value of $n$ can minimize $\mathbb{E}[T_{A+B}]$. We find some results in Theorem~\ref{Lm:RW-opt-n}.

\begin{theorem}
\label{Lm:RW-opt-n}
By using the asymptotic equivalence \eqref{Eq:small}, the expected message passing delay $\mathbb{E}[T_{A+B}]$ reaches the minimum at $n=\sqrt{rk+k}-1$ given $r$ and $k$. Notice that if the above value is not an integer, the optimal $n$ is $\left \lceil{\sqrt{rk+k}-1} \right \rceil$ or $\left \lfloor{\sqrt{rk+k}-1} \right \rfloor$.
\end{theorem}
\begin{proof}
When $k$ and $r$ are given, $\mathbb{E}[T_{A+B}]$ is a function of $n$. Since $\mathbb{E}[T_{A+B}]$ is discrete, we can get its minimum by find an $n^*$ where $\mathbb{E}[T_{A+B}](n=n^*)\le \mathbb{E}[T_{A+B}](n=n^*+1)$ and $\mathbb{E}[T_{A+B}](n=n^*)\le \mathbb{E}[T_{A+B}](n=n^*-1)$. From \eqref{Eq:small}, we have $\mathbb{E}[T_{A+B}(n=i)]\sim_{v}\theta_d v (H_r+H_i-H_{r-i}-H_{i-k})$ and $\mathbb{E}[T_{A+B}(n=i+1)]\sim_{v}\theta_d v (H_r+H_{i+1}-H_{r-i-1}-H_{i+1-k})$.

Define $A_{n=i}=\mathbb{E}[T_{A+B}(n=i+1)]-\mathbb{E}[T_{A+B}(n=i)]$,
then
\begin{small}
\begin{align*}
    A_{n=i}&\sim_{v} \theta_d v (\frac{1}{i+1}+\frac{1}{r-n}-\frac{1}{n+1-k})\\
    &= \frac{\theta_d v}{(i+1)(r-n)(n+1-k)}(n^2+2n+1-rk-k)\\
    &=\frac{\theta_d v}{(i+1)(r-n)(n+1-k)}[(n+1)^2-(rk+k)].
\end{align*}
\end{small}
Since $\frac{\theta_d v}{(i+1)(r-n)(n+1-k)}>0$, $A_{n=i}\ge0$ for $n\ge \sqrt{rk+k}-1$ and $A_{n=i}<0$ for $n< \sqrt{rk+k}-1$. Therefore, $n=\sqrt{rk+k}-1$ minimizes $\mathbb{E}[T_{A+B}]$.
\end{proof}

Observe that for a complete graph with $v$ vertices, the expression for message passing delay becomes $T_{A+B}=v (H_r+H_n-H_{r-n}-H_{n-k})$, as previously derived.

{\it Numerical Analysis:}
\begin{figure}[hbt]
    \centering
    \includegraphics[width=0.46\textwidth]{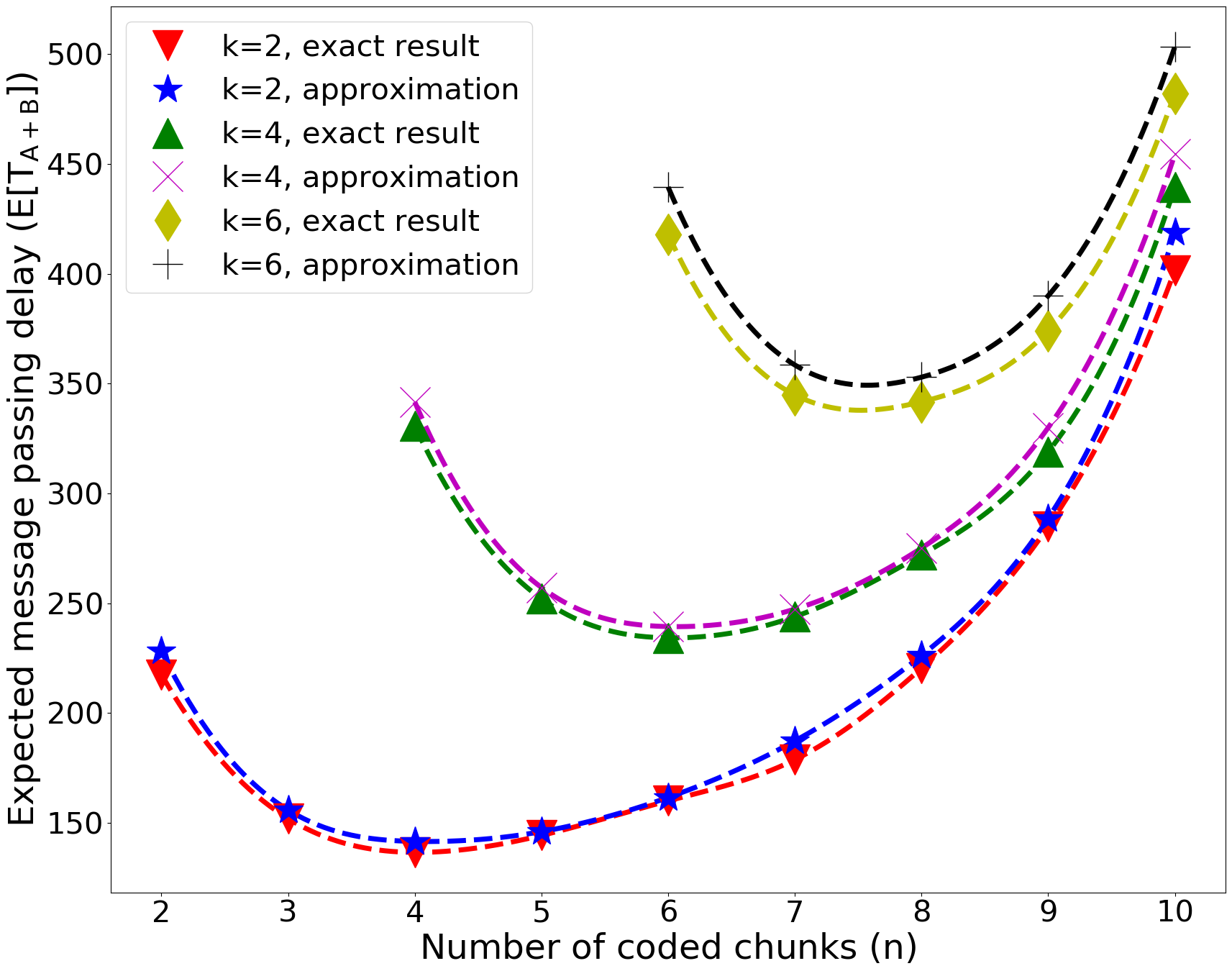}
    \caption{Expected message passing delay $\mathbb{E}[T_{A+B}]$ vs. the number of coded chunks $n$ (cf. \eqref{Eq:small}). This is a regular graph with $100$ vertices and the degree is $5$. The number of relays is $r=10$. Introducing proper data redundancy can reduce the message passing delay and the approximation is close to the exact result. }
    \label{fig:const-com1}
\end{figure}
In Fig.~\ref{fig:const-com1}, we evaluate the expected message passing delay $\mathbb{E}[T_{A+B}]$ vs. the number of coded chunks $n$ for both approximation and exact result. We consider a regular graph with $100$ vertices and the degree is $5$. The approximation is calculated from the expression of $\mathbb{E}[T_{A+B}]$ given in \eqref{Eq:small}, and the exact result is an average of $1000$ sampled delay values.
Some observations are made from the figure: when $n$ is given, $\mathbb{E}[T_{A+B}]$ increases with $k$. For example, when $n=7$, the case "$k=2$" gives the minimum $\mathbb{E}[T_{A+B}]$. When $n$ is large (e.g., $n=10$), the gaps between these three cases are small. These observations are consistent with the theoretical analysis for $k$.   
We also observe that when $k$ is given, $\mathbb{E}[T_{A+B}]$ always reaches the minimum at $n>k$. For example, when $k=4$, the optimal $n$ is $6$ and the code rate $k/n=2/3$. We conclude that introducing proper data redundancy can reduce the message passing delay. This observation is consistent with the result in Theorem~\ref{Lm:RW-opt-n}.
Besides, we observe that for each $k$, the approximation is very close to the exact result. It confirms that \eqref{Eq:small} is good enough to approximate the expected message passing delay.

\subsection{Random Transmission Time}
\label{subsec:Model1cg}
Under the random transmission time model, both the source and collector need to spend some walking time (totally $\eta_1+...+\eta_{S}$) to meet a relay, and then they need to spend the message chunk passing time (defined as $t_{t}$) to delivery and collect the chunk. According to Subsection~\ref{subsec:chucktrans}, $t_{t}$ follows a shifted exponential distribution $\sexpD(\Delta,\lambda)$, where $\Delta=\ell$.

\subsubsection{Dissemination Time}
The source needs to disseminate $n$ coded chunks to $r$($\ge n$) relays. When the source randomly walks on a $d$-regular graph with $v$ vertices, the probability that it meets a relay is $\frac{r}{\theta_d v}$ according to \eqref{Eq:meet}. After the source deposits the first chunk in one of the $r$ relays, the second chunk can only be stored in one of the remaining $r-1$ relays. The probability of meeting an unoccupied relay decreases as the number of occupied relays grows. Therefore, in order to get the dissemination time $T_{A}$, we need to find the time $T_i$ for the source to disseminate the $i^{th}$ ($i=\{1,2,\dots,n\}$) coded chunk. 

\begin{lemma}
\label{Th:dis-time}
For the random transmission time model, the time $T_i$ for the source to disseminate the $i^{th}$ coded chunk to any one of $r-i+1$ relays is
\begin{equation}
    T_i= t_{t}+\sum_{j=1}^{S}\eta_{j} \quad \text{ with probability } (1-p_{r-i+1})^{S-1}p_{r-i+1}.
\end{equation}
Where $S$ is the number of steps the source spends to meet a relay, and $p_{r-i+1}\sim_v\frac{r-i+1}{\theta_{d} v}$ is the probability that the source meets any one of $r-i+1$ relays.
Then we have
\begin{equation}
   \mathbb{E}\left[T_i\right]= \frac{1}{\lambda}+\frac{m}{k}+\frac{\mathbb{E}\left[\eta\right]}{p_{r-i+1}}.
\end{equation}
\end{lemma}
\begin{proof}
Let $p=p_{r-i+1}$, we get the expectation of $T_i$ as follows:
\begin{small}
    \begin{align*}
        \mathbb{E}\left[T_i\right]&=\sum_{S=1}^{\infty}\mathbb{E}\left[t_{t}+\sum_{j=1}^{S}\eta_{j}\right](1-p)^{S-1}p\\
        &=p\mathbb{E}\left[t_{t}\right]\sum_{S=1}^{\infty}(1-p)^{S-1}\\&+p\sum_{S=1}^{\infty}\mathbb{E}\left[\sum_{j=1}^{S}\eta_{j}\right](1-p)^{S-1}.\qquad
    \end{align*} 
\end{small}
Since $\eta_1,...,\eta_{S}$ are i.i.d.,  $\mathbb{E}\left[\sum_{j=1}^{S}\eta_{j}\right]=S\mathbb{E}\left[\eta\right]$. Therefore,
$\mathbb{E}\left[T_i\right]=\mathbb{E}\left[t_{t}\right]+\frac{\mathbb{E}\left[\eta\right]}{p}=\frac{1}{\lambda}+\frac{m}{k}+\frac{\mathbb{E}\left[\eta\right]}{p_{r-i+1}}$.
\end{proof}

Using Lemma~\ref{Th:dis-time}, we can get the expected  total dissemination time $ \mathbb{E}\left[T_{A}\right]$ in the following theorem.
\begin{theorem}
\label{Lm:dis-expect}
The expected dissemination time for the source to transmit all $n$ coded chunks to any $n$ out of $r$ relays is 
\begin{equation}
    \mathbb{E}\left[T_{A}\right]\sim_v\frac{n}{\lambda}+\frac{nm}{k}+\theta_d v\mathbb{E}\left[\eta\right](H_r-H_{r-n}).
\end{equation}
\end{theorem}
\begin{proof}
Since $\mathbb{E}\left[T_i\right]= \frac{1}{\lambda}+\frac{m}{k}+\frac{\mathbb{E}\left[\eta\right]}{p_{r-i+1}}$ for $i \in\{1,2,\dots,n\}$, where $p_{r-i+1}\sim_v\frac{r-i+1}{\theta_{d} v}$, we have
\begin{small}
    \begin{align*}
        \mathbb{E}\left[T_{A}\right]&=\mathbb{E}[\sum_{i=1}^{n}T_i]
        =n\mathbb{E}\left[t_{t}\right]+\sum_{i=1}^{n}\frac{\mathbb{E}\left[\eta\right]}{p_{r-i+1}}\\
        &\sim_v n\mathbb{E}\left[t_{t}\right]+\sum_{i=1}^{n}\frac{\theta_dv\mathbb{E}\left[\eta\right]}{r-i+1}\\
        &=\frac{n}{\lambda}+n\ell+\theta_d v\mathbb{E}\left[\eta\right](H_r-H_{r-n}).\hfill 
    \end{align*}
\end{small}
Since $\ell=\frac{m}{k}$, $\mathbb{E}\left[T_{A}\right]\sim_v\frac{n}{\lambda}+\frac{nm}{k}+\theta_d v\mathbb{E}\left[\eta\right](H_r-H_{r-n})$.
\end{proof}

Notice that  $\eta$ is determined by the distance between two vertices and the speed of the source/collector. To study these two parameters is not the purpose of this paper. Therefore, we will directly assign a value to $\mathbb{E}[\eta]$. Therefore, from Theorem~\ref{Lm:dis-expect}, we see that $\mathbb{E}\left[T_{A}\right]$ increases with $n$ and decreases with increasing $k$. Since $r=\varphi n$ and $k=\xi n$ for some $\varphi,\xi\in(0,1)$, it is not hard to see $\mathcal{O}(\mathbb{E}[T_{A}])=\mathcal{O}(v+n)$.

For a scenario without redundant relays, the number of relays $r$ is equal to the number of coded chunks $n$. Thus, the expected dissemination time is $\mathbb{E}[T_{A}]\sim_{v}\frac{n}{\lambda}+\frac{nm}{k}+\theta_d v H_n$. Namely, we have  $\mathcal{O}(\mathbb{E}[T_{A}])=\mathcal{O}(v\log n+n)$.

\subsubsection{Collection Time}

The collector needs to collect at least $k$ coded chunks to recover the message. Since there are only $n$($\le r$) relays on the graph store the coded chunks, the probability that the collector meets a relay with a coded chunk is $\frac{n}{\theta_d v}$. After the collector retrieves the first chunk in one of the $n$ relays, the second chunk can only be retrieved in one of the remaining $n-1$ relays. 
Therefore, similar to the dissemination time, we find the expected collection time $\mathbb{E}\left[T_{B}\right]$ in the following corollary.

\begin{corollary}
\label{Lm:col-expect}
The expected collection time for the collector to retrieve any $k$ chunks from $n$ relays is 
\[
    \mathbb{E}\left[T_{B}\right]\sim_v \frac{k}{\lambda}+m+\theta_d v\mathbb{E}\left[\eta\right](H_n-H_{n-k}).
\]
\end{corollary}
The above result holds by applying the similar proofs of Lemma~\ref{Th:dis-time} and Theorem~\ref{Lm:dis-expect}.

From Corollary~\ref{Lm:col-expect}, we see that $\mathbb{E}\left[T_{B}\right]$ decreases with increasing $n$ and increases with $k$. Since $k=\xi n$ for some $\varphi,\xi\in(0,1)$, it is not hard to see $\mathcal{O}(\mathbb{E}[T_{B}])=\mathcal{O}(v+k)$.

For a scenario without coding, the code rate $k/n=1$, namely, $k=n$. Thus, the expected collection time is $\mathbb{E}[T_{B}]\sim_{v}\frac{k}{\lambda}+m+\theta_d v H_n$. Namely, we have  $\mathcal{O}(\mathbb{E}[T_{B}])=\mathcal{O}(v\log n+n)$.

\subsubsection{Message Passing Delay}
\label{Sec:UA-joint}
The message passing delay is the sum of the dissemination time and collection time, and its expectation is in Corollary~\ref{Lm:joi-expect}.
\begin{corollary}
\label{Lm:joi-expect}
For the random transmission time model, the expected message passing delay is:
\begin{equation}
\begin{split}
\label{Eq:model1}
    \mathbb{E}\left[T_{A+B}\right]\sim_v &\frac{n+k}{\lambda}+\left(\frac{n}{k}+1\right)m\\&+\theta_d v\mathbb{E}\left[\eta\right](H_r+H_n-H_{r-n}-H_{n-k}).
\end{split}
\end{equation}
\end{corollary}

From the conclusions of the dissemination time and the collection time, we see that there should be an optimal $n$ and an optimal $k$ which minimize the $\mathbb{E}\left[T_{A+B}\right]$. Assume $k$ is given, we find the optimal $n$ in the following theorem.
\begin{theorem}
\label{cor:ran-know}
For the random transmission time model, when $k$ is given, the optimal $n$ appear in the range $[k,\sqrt{rk+k}-1]$.
\end{theorem}
\begin{proof}
Define $A_{n=i}=\mathbb{E}[T_{A+B}(n=i+1)]-\mathbb{E}[T_{A+B}(n=i)]$, we have $A_{n=i}\sim_{v} \frac{1}{\lambda}+\frac{m}{k}+\theta_d v\mathbb{E}\left[\eta\right] (\frac{1}{i+1}+\frac{1}{r-n}-\frac{1}{n+1-k})$. Since $\lambda,m,k>0$, the former terms $\frac{1}{\lambda}+\frac{m}{k}>0$. According to the proof of Theorem~\ref{Lm:RW-opt-n}, we know that the latter term is larger than $0$ when $n>\sqrt{rk+k}-1$, and smaller than $0$ when $n<\sqrt{rk+k}-1$. Therefore, the optimal $n$, as an integer, must appear in the range $[k,\sqrt{rk+k}-1]$.
\end{proof}
From Theorem~\ref{cor:ran-know}, we know that when $\frac{1}{\lambda}+m\gg\theta_d v\mathbb{E}\left[\eta\right]$, it is close to $k$. When $\theta_d v\mathbb{E}\left[\eta\right]\gg\frac{1}{\lambda}+\frac{m}{k}$, it is close to $\sqrt{rk+k}-1$.
 
To analyze the optimal $k$, we assume $n$ is given. From \eqref{Eq:model1}, we see that $k$ appears in the terms $\frac{mn}{k}+\frac{k}{\lambda}-\theta_d v\mathbb{E}\left[\eta\right]H_{n-k}$, where $\frac{mn}{k}$ decreases with increasing $k$, and $\frac{k}{\lambda}-\theta_d v\mathbb{E}\left[\eta\right]H_{n-k}$ increases with $k$. Therefore,  the optimal $k$ changes with the values of system parameters, e.g., $\lambda$, $d$, $v$, $\eta$ and $m$.

Notice that for a complete graph with $v$ vertices, we find the expression of message passing delay is $T_{A+B}=\frac{n+k}{\lambda}+\left(\frac{n}{k}+1\right)m+ v\mathbb{E}\left[\eta\right](H_r+H_n-H_{r-n}-H_{n-k})$. The results in Theorem~\ref{cor:ran-know} still hold.

{\it Numerical Analysis:}
\begin{figure}[hbt]
    \centering
    \includegraphics[width=0.46\textwidth]{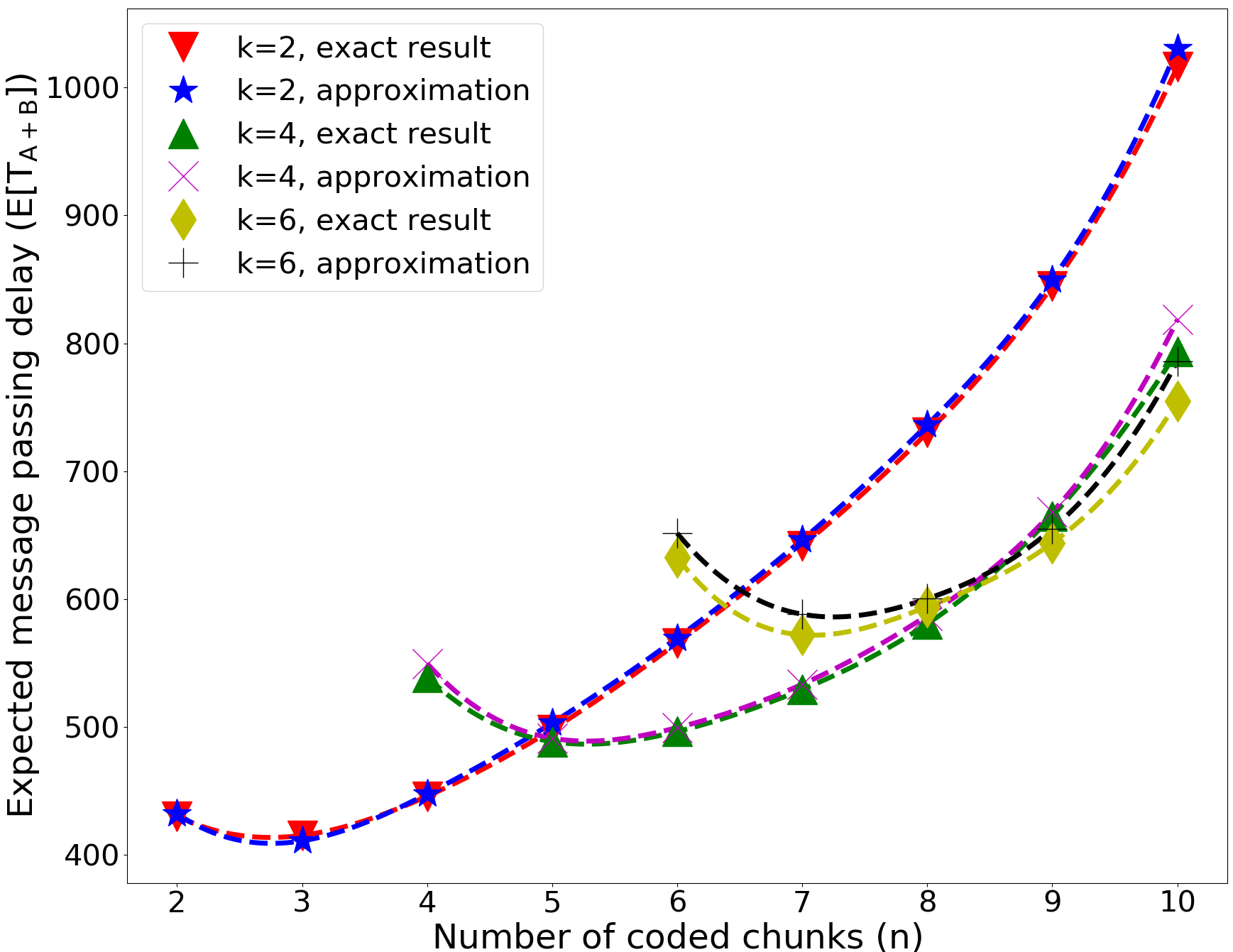}
    \caption{Expected message passing delay $\mathbb{E}[T_{A+B}]$ vs. the number of coded chunks $n$ (cf. \eqref{Eq:model1}). This is a regular graph with $100$ vertices and the degree is $5$. The number of relays is $r=10$ and the message chunk passing time follows $\sexpD(100/k,1)$. Introducing proper data redundancy can reduce the message passing delay and the approximation is close to the exact result.}
    \label{fig:uni-com1}
\end{figure}
In Fig.~\ref{fig:uni-com1}, we evaluate the expected message passing delay $\mathbb{E}[T_{A+B}]$ vs. the number of coded chunks $n$ for both approximation and exact result. We consider a regular graph with $100$ vertices and the degree is $5$. The message chunk passing time $t_{t}$ follows $\sexpD(\ell,\lambda)$, where $\ell=100/k$ is the length of a chunk, $\lambda=1$ is the rate parameter. 
By sampling the shifted exponential distribution, we can get a message passing delay value. The exact result is an average of $1000$ sampled delay values.
The approximation is calculated from the expression of $\mathbb{E}[T_{A+B}]$ given in \eqref{Eq:model1}.
Some observations are made from the figure: when $k$ is given, $\mathbb{E}[T_{A+B}]$ still reaches the minimum at $n>k$. Which indicates that introducing proper data redundancy can reduce the message passing delay. However, comparing to Fig.~\ref{fig:const-com1}, the optimal code rate $k/n$ in Fig.~\ref{fig:uni-com1} is more close to $1$. For example, considering the case "$k=4$", the optimal code rates in Fig.~\ref{fig:const-com1} and Fig.~\ref{fig:uni-com1} are respectively $2/3$ and $4/5$. This observation validates the result in Theorem~\ref{cor:ran-know} that when $\frac{1}{\lambda}+\frac{m}{k}=1+100/k$ is close to $\theta_d v=400/3$, the optimal $n$ appears in the range $[1, \sqrt{11k}-1]$. 

We also observe that when $n$ is given, $\mathbb{E}[T_{A+B}]$ no longer increases with $k$. For example, when $n=6$, the case "$k=4$" has a smaller delay than the case "$k=2$". Although our theoretical analysis also shows this result, we can not describe how the optimal $k$ changes due to there are too many parameters can affect the result. Besides, we observe that the approximation is very close to the exact result. It confirms that \eqref{Eq:model1} is good enough to approximate the expected message passing delay.

\subsection{Comparisons and Conclusions}
In Table~\ref{Table:complexity}, we compare different scenarios with or without redundant relays (RR) and coding. We conclude that redundant relays can help to decrease the dissemination time and coding can decrease the collection time. For example, considering the constant transmission time model, with the help of redundant relays and coding, the average dissemination and collection time decrease from $\mathcal{O}(v\log n)$ to $\mathcal{O}(v)$, respectively. 
\begin{table*}[hbt]
\begin{center}
\begin{threeparttable}
    \caption{Dissemination and collection times with or without redundant relays and coding}
    \label{Table:complexity}
		\begin{small}
			\begin{tabular}{@{}lll|ll@{}}
				\toprule 
				&  \multicolumn{2}{c@{}}{\sc Constant TT\tnote{1}} &  \multicolumn{2}{c@{}}{\sc Random TT}
				\\\cmidrule{2-5}
				& {\bf Dissemination} & {\bf Collection} & {\bf Dissemination} & {\bf Collection} \\ [2mm]
			    {\bf w/ RR\tnote{1} and coding}  & $\mathcal{O}(v)$ & $\mathcal{O}(v)$ & $\mathcal{O}(v+ n)$  &  $\mathcal{O}(v+ k)$ \\[2mm] 
				{\bf w/o RR and w/ coding}  & $\mathcal{O}(v\log n)$   &  $\mathcal{O}(v)$ & $\mathcal{O}(v\log n+n)$ & $\mathcal{O}(v+k)$\\[2mm]
				{\bf w/ RR and w/o coding} & $\mathcal{O}(v)$ & $\mathcal{O}(v\log n)$ & $\mathcal{O}(v+ n)$ & $\mathcal{O}(v\log n+n)$ \\[2mm]
				{\bf w/o RR and coding}  & $\mathcal{O}(v\log n)$ & $\mathcal{O}(v\log n)$ & $\mathcal{O}(v\log n+n)$ & $\mathcal{O}(v\log n+n)$\\
				\bottomrule
			\end{tabular}
		\end{small}
\begin{tablenotes}
\footnotesize
\item[1] TT is transmission time (cf.~\ref{subsec:chucktrans}).
\item[2] RR is redundant relays.
\end{tablenotes}
    \end{threeparttable}
	\end{center}
\end{table*}

In Table~\ref{Table:delay}, we summarize the results of the expected message passing delay in both constant transmission time and random transmission time models. In the constant transmission time model, we deduce the optimal $k$ and $n$ from the expression of $\mathbb{E}[T_{A+B}]$. However, in the random transmission time model, the exact values of optimal $k$ and $n$ are hard to deduce from the expression of $\mathbb{E}[T_{A+B}]$.
\begin{table*}[hbt]
\begin{center}
\begin{threeparttable}
    \caption{Conclusions of important results in both delay models}
    \label{Table:delay}
		\begin{small}
			\begin{tabular}{@{}lllll@{}}
				\toprule 
				& &  $\bf{\mathbb{E}[T_{A+B}]}$ 
				& {\bf Optimal $k$ (given $n$)}
				& {\bf Optimal $n$ (given $k$)}
				\\ [2mm]
			    \multirow{4}{*}{\rotatebox{90}{
						{\sc{Model}}}}  & 
				{\bf Constant TT\tnote{1}} & $\theta_d v(H_r+H_n-H_{r-n}-H_{n-k})$ & $k=1$ 
				 &  $\left \lceil{\sqrt{rk+k}-1} \right \rceil$ or $\left \lfloor{\sqrt{rk+k}-1} \right \rfloor$\\[2mm] 
				& {\bf Random TT}   &  $\begin{aligned}&\frac{n+k}{\lambda}+\left(\frac{n}{k}+1\right)m+\\& \theta_d v\mathbb{E}\left[\eta\right](H_r+H_n-H_{r-n}-H_{n-k})\end{aligned}$ & not clear & $k\le n\le \sqrt{rk+k}-1$\\[2mm]
				\bottomrule
			\end{tabular}
		\end{small}
	\begin{tablenotes}
\footnotesize
\item[1] TT is transmission time (cf.~\ref{subsec:chucktrans}).
\end{tablenotes}
    \end{threeparttable}
	\end{center}
\end{table*}

\subsection{Source and Collector's Starting Time}
\label{subsec:starting}
In the previous subsections, we only consider the message passing delay which defined as the sum of the dissemination time and collection time. In practice, we should also look into the source and collector's dynamics on when to start each process. Thus, the overall message passing time, defined as the time from the beginning of the source's dissemination to the end of the collector's collection, is also an important performance metric for the message passing system. 

To analyze the overall message passing time, in general, two scenarios should be considered: 1) the collector starts collecting chunks after the dissemination. Thus, the overall message passing time is larger than the message passing delay, since there will be some intermediate time between dissemination and collection. Nevertheless, the previous conclusions also hold for the overall message passing time analysis, because the intermediate time is unrelated to the dissemination and collection times. 2) The collector starts collecting chunks during the dissemination. Thus, the overall message passing time is smaller than the message passing delay and the previous conclusions do not always hold for the overall message passing time analysis. This is an interest problem that needs to be studied in the future.


\subsection{Multiple Sources and Collectors with Network Coding for Further Reduction of Delay}
\label{sub:multi}
With the help of redundant relays and coding, we reduced the dissemination and collection time for both the constant and random transmission time models. However, the effectiveness of redundant relays and coding is limited. For example, we can not reduce the average dissemination time beyond $\mathcal{O}(v)$ under the constant transmission time model. Nevertheless, we can further reduce the delay by introducing multiple sources and collectors. We outline the main ideas below and leave the details of this problem for a follow-up study.

Based on the results of \cite{alon2008many}, we know that $w$ random walks on a complete graph or a $d$-regular graph achieve $\mathcal{O}(w)$ reduction of the single walk graph cover time.  (This is true for some other graphs \cite{elsasser2011tight} as well.) Thus multiple sources will reach the relays faster than a single source but
need to disseminate different message chunks to reduce the message passing delay. Thus, each source must know which chunks have been disseminated by the other sources. Such assumption is impractical and inadequate for covert communications scenarios. Without this assumption, having multiple sources and collectors may even worsen the delay. To see that, consider two sources that disseminate a message of $2$ data chunks to $4$ relays by randomly selecting a chunk each time they encounter a relay. If both sources disseminate only one data chunk, the dissemination time is significantly shorter than when one source disseminates $2$ data chunks. However, the two sources will disseminate the same chunk with a probability $0.5$. If both sources disseminate $2$ data chunks, the dissemination time is larger than the one source scenario.

The sources can apply various strategies to ensure that they disseminate different chunks. For example, if there are two sources, one can disseminate the even chunks and the other odd chunks. 
However, we will have higher benefits if we do not fix the number of chunks that each source has to disseminate as long as the sum of the disseminated chunks is $n$.
Moreover, the splitting strategy does not reduce the collection delay.
We may use network coding to improve the multiple sources/collectors and reduce the collection delay. In this way, instead of sending the data chunks to the relays, the sources will send a random linear combination of the data chunks. 

For the most general scenario, the sources may not even be able to agree or know which chunks are disseminated before/during the dissemination.  Here, using network coding is essential. Each time a source encounters a relay, it randomly generates a coded chunk as a linear combination of $k$ data chunks. Thus, the sources only need to disseminate $n$ coded chunks together. 
In the collection phase, the collectors need to collect totally $k^{'}$ out of $n$ chunks, where $k^{'}$ is just slightly larger than or equal to $k$. 
Then, when the field size is sufficiently large, the $k$ data chunks can be decoded from $k{'}$ coded chunks with high probability.
Using fountain codes, e.g., LT \cite{luby2002lt} or Raptor codes \cite{shokrollahi2006raptor}, instead of random network codes can simplify decoding. 

\section{Covert Message Passing}
\label{Sec:covertpass}
We now consider message passing where in addition to Alice, Bob, and relays,  there is another communication participant, Willie the warden. We are interested in the probability that the message passing from Alice and Bob is hidden from Willie who has certain mobility and detection capabilities.

The transmission stays covert only when Alice successfully passes a message to Bob through some relays without being detected by Willie. Although covert communication has been studied in many papers, covert message passing in mobile networks is recently proposed. Therefore, there is no model on how the warden detects. In the following, we provide two possible models. 

\subsection{Two Warden Detection Models}
\label{Sec:warden}
\subsubsection*{Random Patrolling Model} In this  model, see Fig.~\ref{fig:wardenmodel} (left), the warden Willie walks randomly on the same graph with Alice and Bob.  We assume that all mobile participants are moving synchronously. To establish that Alice and Bob are communicating, Willie has to 1) meet one of them at a relay and 2) detect the transmission to the relay is taking place before it was over. The more time Willie has to observe a chunk transmission, the higher his probability of detection will be. The detection probability is thus an increasing function of the chunk length $\ell$. However, decreasing the chunk length will result in having to pass more  chunks, which in turn gives Willie more ($m/\ell$) opportunities to observe a transmission.

\begin{figure}[hbt]
    \centering
    \includegraphics[width=0.48\textwidth]{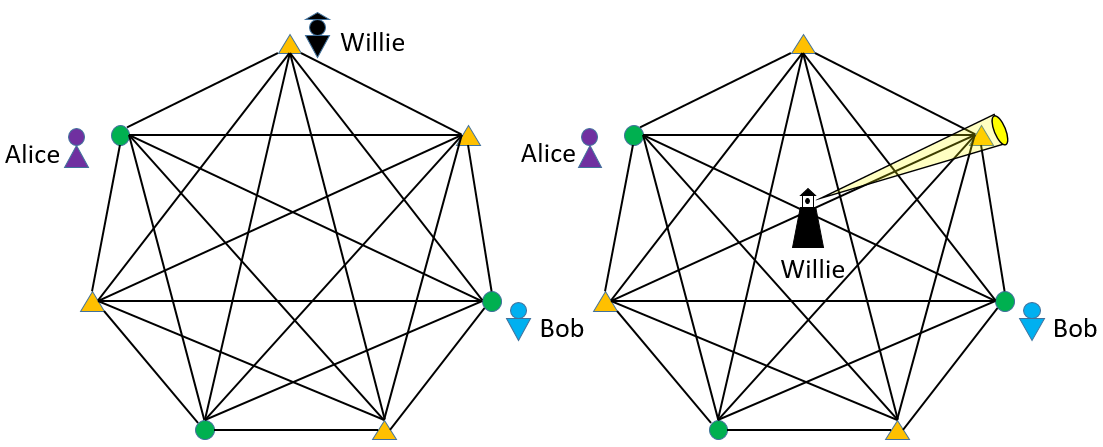}
    \caption{Two warden models are proposed: random walk (left) and uniform surveillance (right). For the random walk model, Willie walks randomly on the graph with Alice and Bob. He can only detect the communication when he meets Alice or Bob at the relay. For the uniform surveillance model, Willie monitors each vertex uniformly. He can only detect the communication when he monitors the relay during Alice/Bob's transmission.}
    \label{fig:wardenmodel}
\end{figure}

\subsubsection*{Uniform Surveillance Model}
In this model, we base covertness on the assumption that the warden Willie can monitor part of the vertices for some given time. A straightforward and informal way to visualize this model is to imagine that the warden is stationed somewhere ``in the middle'' of the graph\footnote{This is a very informal statement. There is no need for an exact ``middle'', we just need the warden to be at a place where he can observe different parts of the area at different times.}, on top of a lighthouse, see Fig.~\ref{fig:wardenmodel} (right). This way, he can only check the part where the lighthouse sheds its light and can not see what is happening behind him. We can also imagine the warden applies a Round-Robin protocol to check each vertex of the graph. Thus, he can monitor vertices uniformly.

To formalize this model, we assume that Alice and Bob transmit data to relays without implementing any covertness scheme. Therefore, if the warden happens to check a graph vertex while the data chunk transmission is taking place, he will detect it with probability $1$. We further assume that when Alice or Bob starts transmitting on a relay, the warden's time of arrival at this relay follows a uniform distribution. Thus, if the warden monitors the relay during the transmission, he will detect it; otherwise, he will not. This model was first proposed in paper \cite{peng2019straggling} which provided the analysis for a complete graph.

\subsection{Performance Metrics}
\subsubsection{Covertness Probability}
\label{subsub:covert}
Covertness probability is defined as the probability that Alice transmits a message to Bob without being detected by Willie. For example, assume that the message has $2$ data chunks, Alice needs to transmit $2$ times to $2$ relays and Bob also needs $2$ times to collect the chunks. If during each time, Willie will detect the transmission with a probability $P_d$, the covertness probability is $P_c=(1-P_d)^4$. Notice that when Willie detects the transmission, it does not mean he will get the content of the message. The message may be camouflaged as noise to avoid detection. However, this is another covert communication problem and will not be studied in this paper.


\subsubsection{Covertness vs.\ Delay Tradeoff}

To optimize the covertness probability, we will also consider using splitting and data redundancy. From the definition of covertness probability, we know that although splitting increases the number of chances that Willie detects, it also decreases the detection probability $P_d$ for each chance (data chunk is smaller than the message). Therefore, it is hard to tell whether splitting decreases the covertness probability. 
Meanwhile, it is obvious that introducing data redundancy decreases the covertness probability. Therefore, 
revisit our message passing delay models, the tradeoff between the covertness probability and delay becomes apparent here. 

On the one hand, if Alice delivers the data chunks without data redundancy to relays, then the probability of detection is small. This happens because the chance that Willie ``sees'' her is inversely proportional to the number of nodes. On the other hand, the delay is increased because Bob will have to visit many nodes until he meets all the relays that hold Alice's data chunks.
At the other side of the spectrum, if Alice encodes the data chunks and delivers the coded chunks to more relays, then it will take Bob fewer steps to retrieve it. Thus, the message passing delay may reduce (see conclusions in Sec.~\ref{Sec:messagedelay}). However, the probability that they are caught increases significantly because of the more times that Alice has to deliver to the relays.

\subsection{Notations}
\begin{center}
\begin{small}
    \begin{tabular}{rcl}
       $P_d$ & - & detection probability for each transmission\\
       $P_c$ & -  & covertness probability\\
       $\beta$  & -  & the number of wardens\\
       $t_t$ & - & the message chunk passing time\\
       $t_a$ & - & the warden's arrival time
    \end{tabular}
    \end{small}
\end{center}
The parameters $\beta$, $t_t$ and $t_a$ are the system parameters. $P_d$ and $P_c$ are the performance metrics.

\section{Covertness Probability Analysis}
\label{Sec:analysis_cp}

In Sec.~\ref{subsub:covert}, we introduced the covertness probability $P_c$ and provided an example to calculate $P_c$ given the transmission detection probability $P_d$. In practice, $P_d$ is not always a constant. The warden Willie will detect the transmission with a higher probability when he "sees" the transmission for a longer time. In another word, $P_d$ increases with the length of a chunk $\ell$. In this section, we provide the expression of $P_d$ as a function of $\ell$ for two different warden models, respectively. Meanwhile, we provide the general expression of $P_c$.

\subsection{Transmission Detection Probability}
\label{sec:covert}
\subsubsection{Random Patrolling Model}
\label{subsec:random}
In this model, $\beta$ wardens, a source, and a collector walk randomly on a regular graph. As we introduced in Sec.~\ref{Sec:warden}, the warden firstly needs to meet the source or collector, and then detects the transmission.
For a $d$-regular graph with $v$ vertices, according to \eqref{Eq:slot}, the probability that the warden meets the source/collector at a relay is $\frac{\beta}{\theta_d v}$. After meeting the source or collector, we consider two scenarios for the warden: a) the warden will detect the transmission immediately. Thus, the detection probability is
\begin{equation}
\label{Eq:ptd1}
    P_{d}=\frac{\beta}{\theta_d v}. 
\end{equation}
The above formula shows that $P_{d}$ only relates to the degree of graph $d$ and the number of wardens $\beta$, thus we consider the detection probability as a constant.

b) The warden will fail to detect with some probability. In practice, the warden can not always detect the transmission successfully for some reasons, e.g., the source tries to hide the message, the noise of the detection channel, etc. It is intuitively to imagine that the longer time for the source or collector to transmit the coded chunk, the higher probability the warden detects successfully. Since the transmission time relates to the length of the coded chunk, here, we simply assume the detection probability is a linear function of $\ell$.
\begin{equation}
\label{Eq:ptd2}
    P_{d}(\ell)=\frac{\ell\beta}{m\theta_d v}.
\end{equation}
Since $\ell=\frac{m}{k}$, we see that $P_d$ is a function of $k$, and it decreases as $k$ increases. 

\subsubsection{Uniform Surveillance Model}
\label{subsec:uniform}
In this model, when a source or a collector arrives at a relay and starts transmitting the data chunk, the warden will arrive (monitor) at this relay following a uniform distribution $U(0,W)$. As we introduced in Sec.~\ref{Sec:warden}, whether the warden will detect the transmission successfully depends on the warden's arrival time $t_a$ and the source/collector's message chunk transmission time $t_{t}$.
We consider the warden's arrival time $t_a$ follows a uniform distribution, i.e., $t_a\sim U(0,W)$. And the message chunk passing time $t_{t}$ follows a shifted exponential distribution i.e., $t_{t}\sim\sexpD(\Delta, \lambda)$ (the tail is given as $\Pr\{t_{t} > x\} = e^{-\lambda(t_t-\Delta)}$ for $t_t > 
\Delta$). Where $\Delta=\ell$ indicates the time to transmit a chunk and the exponential tail is some inherent additive system randomness at each relay, which does not depend on the chunk length $\ell$. 
Thus, the detection probability $P_{d}$ is given by the following theorem:
\begin{theorem}
\label{Th:covert1}
For the uniform surveillance model, 
the probability that the warden arrives during the transmission (i.e., detects the transmission) is 
\begin{equation}
\label{Eq:detect2}
        P_d(\ell)=\begin{cases}
             \frac{1}{\lambda W}+\frac{\ell}{W}-\frac{e^{-\lambda (W-\ell)}}{\lambda W} & \text{for }  W \ge \ell \\
               \hfil 1  & \text{for }  W < \ell
\end{cases}
\end{equation}
\end{theorem}
\begin{proof}
Since $t_{t}\sim \sexpD(\ell,\lambda)$, it is obvious that $t_{t}\ge \ell$. We assume the warden's arrival time is $t_{a}\sim U(0,W)$, then we have $t_{a}\le W$.

If $W<\ell$, the warden will definitely arrive before the transmission is done, the detection probability is $P_d=1$. 
\\[1ex]
If $W> \ell$, we calculate the detection probability as follows:
\begin{small}
    \begin{align*}
        P_d(\ell)&=P(t_{t}\ge t_{a})=\int^{W}_{0}\int^{\infty}_{t_{a}}f_{t}(t_{t})f_{a}(t_{a}) \ dt_{t} \ dt_{a}\\
        &=\int^{W}_{\ell}\int^{\infty}_{t_{a}}f_{t}(t_{t}) \ dt_{t} \ f_{a}(t_{a}) \ dt_{a}\\&+\int^{\ell}_{0}\int^{\infty}_{\ell}f_{t}(t_{t}) \ dt_{t} \ f_{a}(t_{a}) \ dt_{a}\\
        &=\int^{W}_{\ell} e^{-\lambda (t_t-\ell)}f_{a}(t_{a}) \ dt_{a}+\frac{\ell}{W}\\
        &=\int^{W}_{\ell} \frac{1}{W} e^{-\lambda (t_t-\ell)} \ dt_{a}+\frac{\ell}{W}
        =\frac{1}{\lambda W}+\frac{\ell}{W}-\frac{e^{-\lambda (W-\ell)}}{\lambda W}.
    \end{align*}
\end{small}
\end{proof}

\subsection{Covertness Probability}
\label{sub:covertprob}
The communication between the source and collector stays covert only when all the coded chunk transmissions are undetected. Recall that the source needs to disseminate the $n$ coded chunks and the collector needs to collect $k$ coded chunks. Therefore, the total number of coded chunk transmissions is $n+k$ and the covertness probability is 
\begin{equation}
     P_c(k)= (1-P_d)^{n+k}.
 \label{eq:covert1}   
\end{equation}

When $k$ is given, since $P_d$ and $n$ are independent and $0\le P_d\le1$, the covertness probability $P_c$ decreases with increasing $n$.
Since $P_d$ is a function of $\ell$ and $\ell=\frac{m}{k}$, $P_c$ is a function of $k$. Thus, when $n$ is given, we need further analyze the optimal $k$ for each detection model.

For the random patrolling model, $P_d$ is a constant under the worst-case (see \eqref{Eq:ptd1}). Otherwise, $P_d$ increases linearly with $\ell$ (see \eqref{Eq:ptd2}). Thus, we can easily get the expression of $P_c$ by substituting \eqref{Eq:ptd1} or \eqref{Eq:ptd2} to \eqref{eq:covert1}. From the expressions, we find the maximum of the covertness probability in Theorem~\ref{Lm:random}.  
\begin{theorem}
\label{Lm:random}
For the random patrolling model, given the number of coded chunks $n$, the covertness probability $P_c$ changes as follows:
\begin{enumerate}
      \item Under the worst-case, $P_d=\frac{\beta}{\theta_d v}$ is a constant. Thus, for a given $n$, $P_c$ decreases with increasing $k$, which means $P_c$ reaches the maximum at $k=1$.
      \item Otherwise, $P_d(k)=\frac{\beta}{k\theta_d v}$ is a function of $k$. Thus, $P_c$ increases with $k$, which means $P_c$ reaches the maximum at $k=n$.
\end{enumerate}
\end{theorem}
\begin{proof}
For property 1), since $P_d$ is a constant and $0\le 1-P_d<1$, it is obvious that $P_c$ decreases with $k$. Since $k\ge 1$ is an integer, the maximum covertness probability is at $k=1$, i.e., $P_c=(1-P_d)^{n+1}$.

For property 2), we know $1\le k\le n$ is an integer. Let $a\in[1,n)$ be an integer. To prove $P_c(k)$ increases with $k$, we only need to show $P_c(k=a+1)>P_c(k=a)$ for all $a$.

According to \eqref{Eq:ptd2} and $\ell=\frac{m}{k}$, we have  $P_d(k)=\frac{\beta}{k\theta_d v}$. Define $B=\frac{\beta}{\theta_d v}$, then $P_d(k)=\frac{B}{k}$. Therefore, we only need to show $(1-B/(a+1))^{n+a+1}>(1-B/a)^{n+a}$.

Notice that $0<B\le 1$ and $1\le a<n$, we can transform the inequality into $\frac{a}{a-B}>\frac{a+1}{a+1-B}\sqrt[n+k]{\frac{a+1}{a+1-B}}$. Thus, we get $ \frac{a(a+1-B)}{(a-B)(a+1)}>\sqrt[n+k]{\frac{a+1}{a+1-B}}$.

For the left term, we have $\frac{a(a+1-B)}{(a-B)(a+1)}=1+\frac{B}{(a-B)(a+1)}>1+\frac{B}{(a-B+1)(a+n)}=\frac{1}{n+a}(\frac{a+1}{a+1-B}+n+a-1)$.

According to the arithmetic mean-geometric mean inequality,
$\frac{1}{n+a}(\frac{a+1}{a+1-B}+n+a-1)>\sqrt[n+k]{\frac{a+1}{a+1-B}}$.

Finally, for any integer $1\le a<n$, we have $P_c(k=a+1)>P_c(k=a)$.
\end{proof}

For the uniform surveillance model, since $P_d$ is a function of $k$, $P_c$ is also a function of $k$. From Theorem~\ref{Th:covert1}, the term $\frac{\ell}{W}$($=\frac{m}{kW}$) decreases as $k$ increases, and the term $-\frac{e^{-\lambda (W-\ell)}}{\lambda W}$($=-\frac{e^{-\lambda (W-m/k)}}{\lambda W}$) increases with $k$. When $\lambda$ is sufficiently large, the former term is much larger than the latter term. Thus, $P_d$ also decreases as $k$ increases. Since $P_c$ decreases as $k$ increases when $P_d$ is given, there exists a tradeoff for $k$ to maximize $P_c$. When $\lambda$ is sufficiently small, the former term is smaller than the latter term. Thus, $P_d$ increases with $k$, which leads to $P_c$ reaches its maximum at $k=1$.

\section{Delay vs. Covertness Probability Tradeoffs}
\label{Sec:tradeoff}
In the previous sections, we know that introducing data redundancy reduces both the message passing delay and the covertness probability. Thus, it is important to study how much redundancy affects the delay and probability tradeoffs. 

\subsection{Constant Transmission Time vs. Random Patrolling}
\label{Sec:rw}
Since both the random patrolling model and the constant transmission time model do not relate to the message chunk passing time, we combine these two models to analyze the tradeoff between the covertness probability $P_c$ and the expected message passing delay $\mathbb{E}[T_{A+B}]$.

\begin{figure}[hbt]
    \centering
    \includegraphics[width=0.46\textwidth]{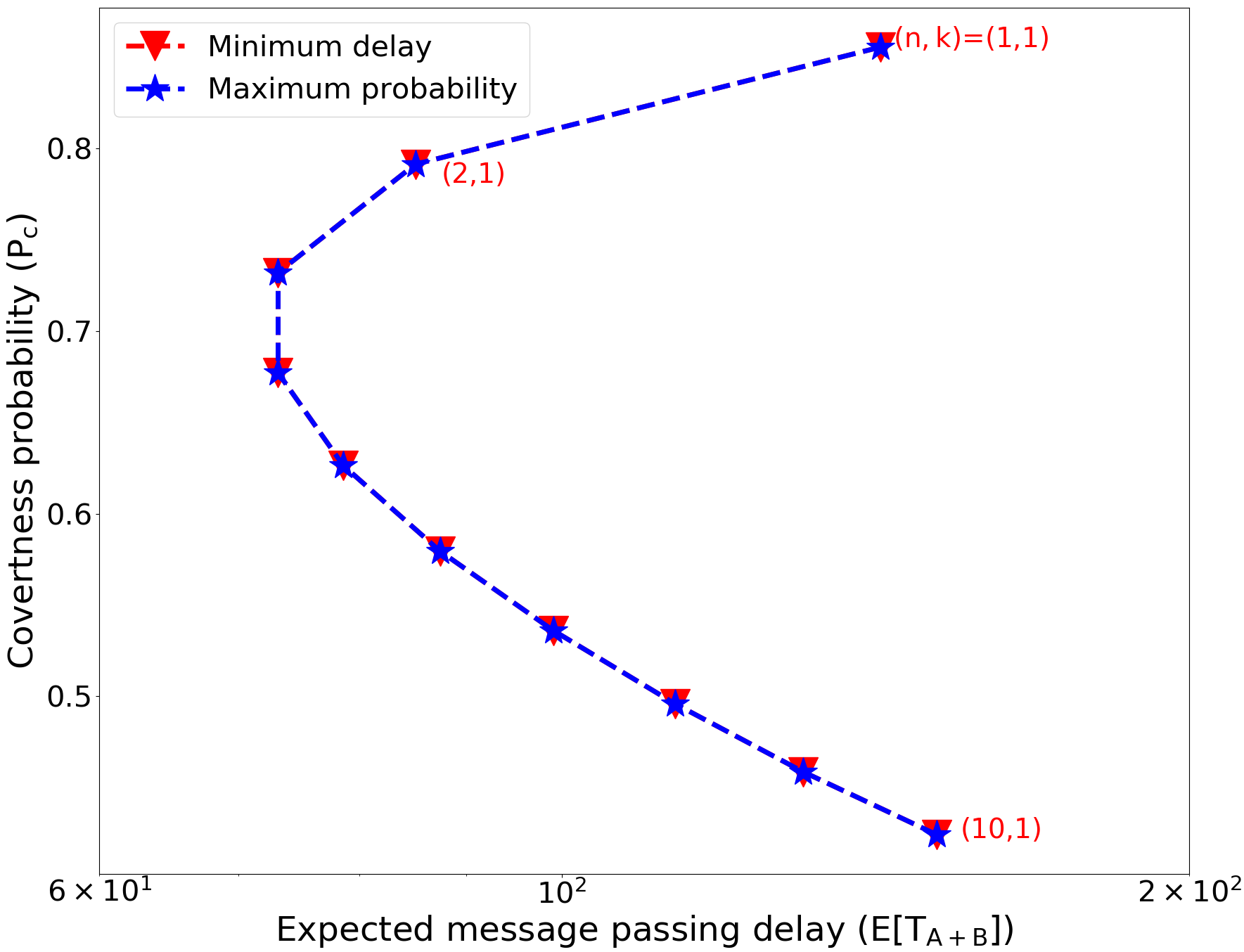}
    \caption{Covertness probability $P_c$ (cf. \eqref{Eq:small}) vs. the expected number of steps $\mathbb{E}[T_{A+B}]$ (cf. \eqref{eq:covert1}) as $n$ increases from $1$ to $10$. We consider the worst-case, i.e., the detection probability is a constant (cf. \eqref{Eq:ptd1}). This is a regular graph with $100$ vertices and the degree is $5$. The number of wardens is $\beta=10$ and the number of relays is $r=15$. There exists an $(n,k)$ which simultaneously maximizes $P_c$ and minimizes $\mathbb{E}[T_{A+B}]$.}
    \label{fig:rw-worst}
\end{figure}
In Fig.~\ref{fig:rw-worst}, we evaluate \eqref{Eq:small} and \eqref{eq:covert1} to see the tradeoff between the detection probability and the transmission delay for the worst case warden's detection probability $P_{d}=\frac{\beta}{\theta_d v}$. We consider a regular graph with $v=100$ vertices and the degree is $d=5$. There are $15$ relays distributed uniformly on vertices of the graph. $10$ wardens, $1$ source and $1$ collector walk randomly on the graph. We evaluate $P_c$ vs. $\mathbb{E}[T_{A+B}]$ as the number of coded chunks $n$ increases from $1$ to $10$. Each point in the figure is a different case with different values of $(n,k)$. Two different strategies are considered: minimum delay and maximum probability. For the minimum delay, we firstly find the optimal $k$ which minimizes $\mathbb{E}[T_{A+B}]$ given each $n$. Then we calculate $P_c$ by using the same $k$. For the maximum probability, on the contrary, we firstly find the optimal $k$ for $P_c$ given each $n$, and then calculate $\mathbb{E}[T_{A+B}]$. From Fig.~\ref{fig:rw-worst}, we observe that the results for both strategies coincide with each other. It means that the optimal $k$ simultaneously maximizes $P_c$ and minimizes $\mathbb{E}[T_{A+B}]$. 
We also observe that $P_c$ decreases with increasing $n$ and $\mathbb{E}[T_{A+B}]$ reaches the minimum at $n=4$. 
These observations are consistent with Theorems~\ref{Lm:RW-opt-n} and \ref{Lm:random}. Finally, we conclude that $k=1$, i.e., no splitting, is the overall optimal strategy. Replication ($n>1$ and $k=1$) can reduce the transmission delay, but it also decreases the covertness probability.

\begin{figure}[hbt]
    \centering
    \includegraphics[width=0.46\textwidth]{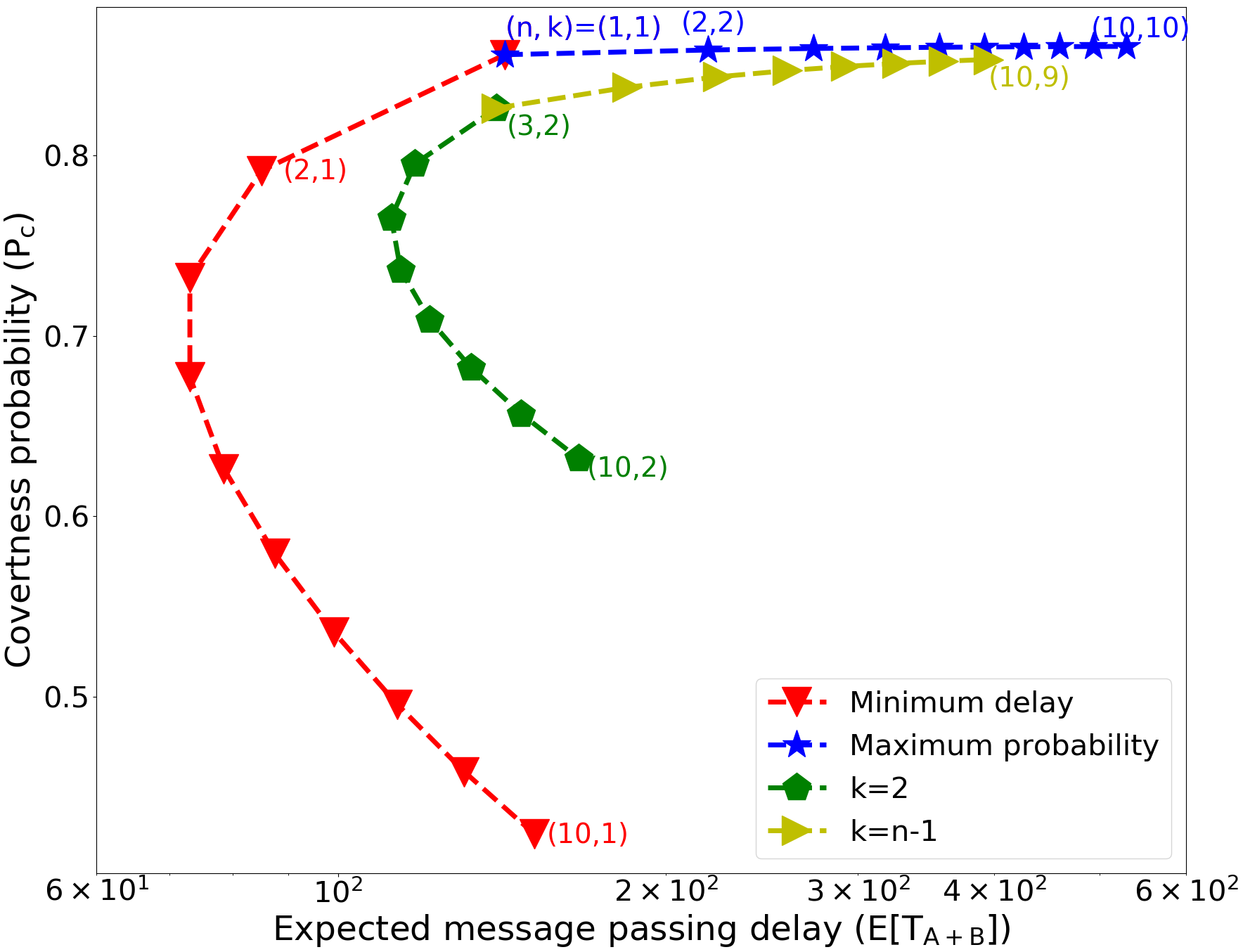}
    \caption{Covertness probability $P_c$ (cf. \eqref{Eq:small}) vs. the expected message passing delay $\mathbb{E}[T_{A+B}]$ (cf. \eqref{eq:covert1}) as $n$ increases from $1$ to $10$. We consider the detection probability is a function of $\ell$ (cf. \ref{Eq:ptd2}). This is a regular graph with $100$ vertices and the degree is $5$. The number of wardens is $\beta=10$ and the number of relays is $r=15$. It is impossible to reach the maximum $P_c$ and the minimum $\mathbb{E}[T_{A+B}]$ simultaneously.}
    \label{fig:rw-linear}
\end{figure}
In Fig.~\ref{fig:rw-linear}, we consider the warden's detection probability increases linearly with the length of data chunk $\ell$, i.e., $P_{d}=\frac{\ell\beta}{\theta_d v}$ (where $\ell=\frac{m}{k}$). The same regular graph is adopted from Fig.~\ref{fig:rw-worst}. We evaluate $P_c$ vs. $E[T_{A+B}]$ as $n$ increases from $1$ to $10$. Four different strategies are considered: minimum delay, maximum probability, "$k=2$" and "$k=n-1$". For the minimum delay, we have the same observations as them in Fig.~\ref{fig:rw-worst}. This is because the expression of $E[T_{A+B}]$ is unchanged, and it still reaches minimum at $(n,k)=(4,1)$. For the maximum probability, $P_c$ reaches minimum at $k=n$, which means splitting is optimal. It is obvious that the curves of the minimum delay and the maximum probability are far from each other. Therefore, we can not simultaneously decrease the message passing delay and increase the covertness probability. In the figure, we also see the results for "$k=2$" and "$k=n-1$" strategies. When we have a minimum requirement for $P_c$ or $E[T_{A+B}]$, then $k=2$ or $k=n-1$ strategies may provide a better performance. For example, if we require $P_c>0.8$, then the case $(n,k)=(4,2)$ gives a relatively smaller $E[T_{A+B}]$.

\subsection{Random Transmission Time vs. Uniform Surveillance}
\label{Sec:ua}
Since both the uniform surveillance detection model and the random transmission time model relate to the message chunk passing time, we combine these two models to analyze the tradeoff between the covertness probability $P_c$ and the expected message passing delay $\mathbb{E}[T_{A+B}]$.

\begin{figure*}[hbt]
    \centering
    \includegraphics[width=0.474\textwidth]{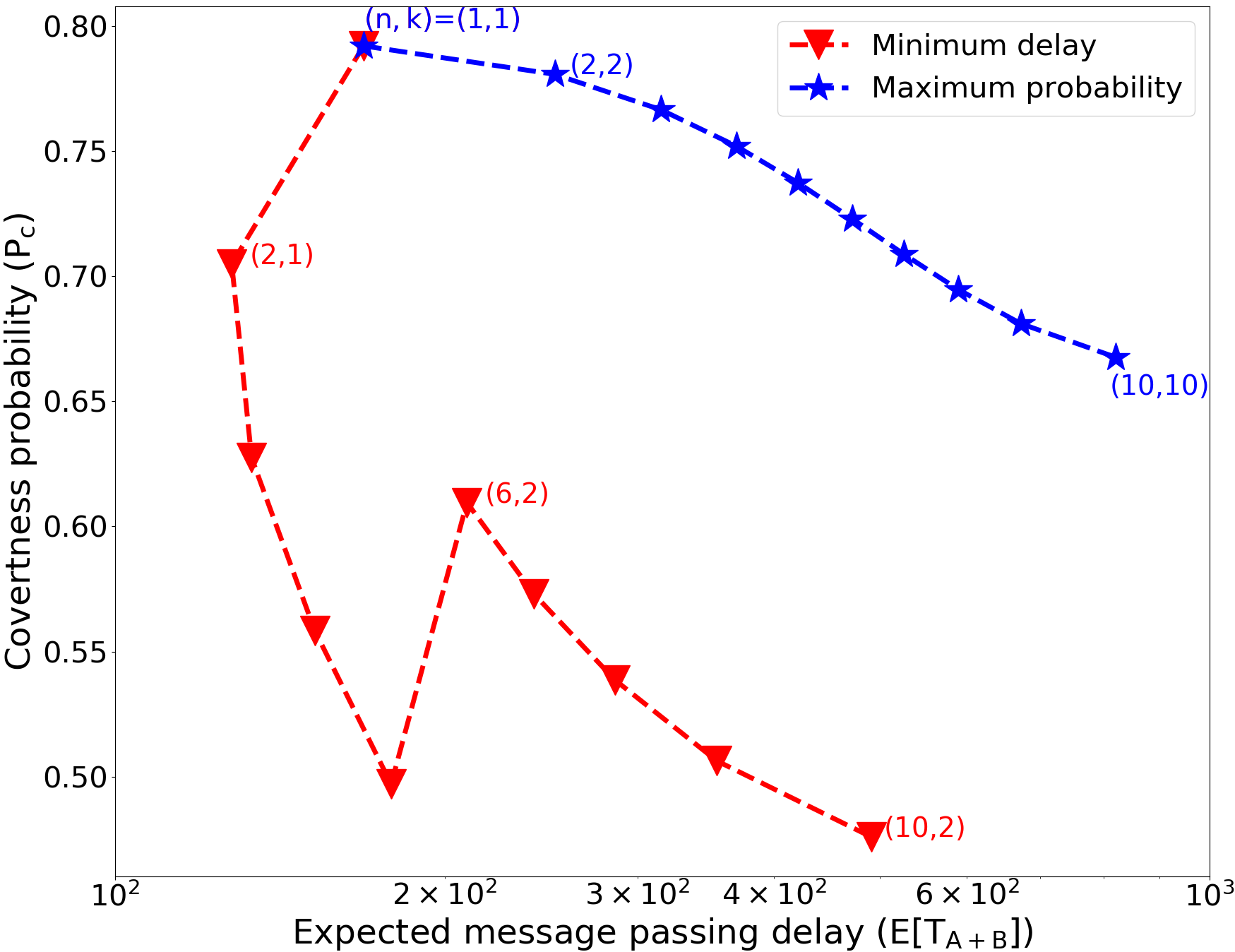}
    \includegraphics[width=0.46\textwidth]{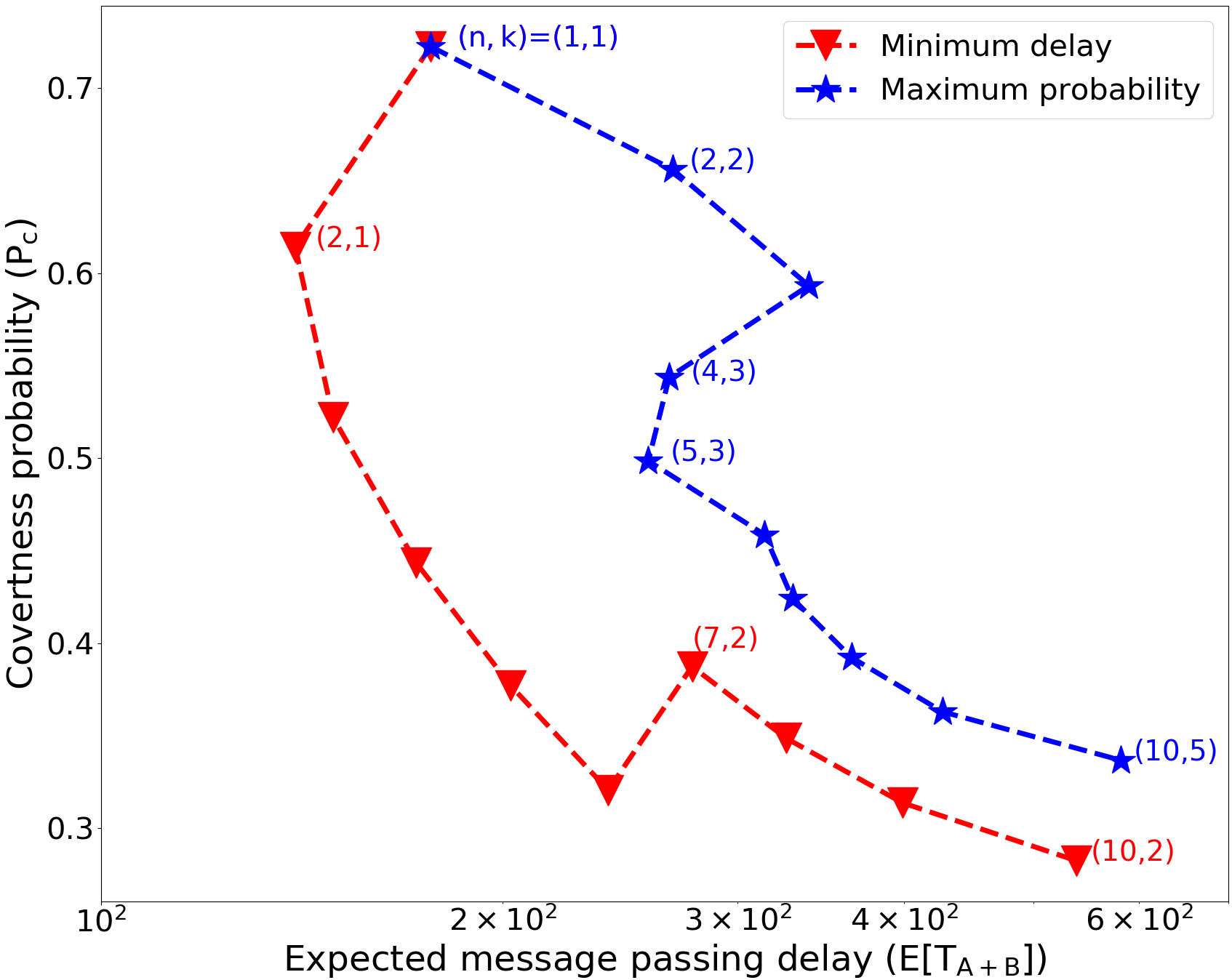}
    \caption{Covertness probability $P_c$ (cf. \eqref{Eq:model1}) vs. the expected message passing delay $\mathbb{E}[T_{A+B}]$ (cf. \eqref{Eq:detect2} and \eqref{eq:covert1}) as $n$ increases from $1$ to $10$. The left subfigure has a message chunk passing time $t_{t}\sim\sexpD(10/k,1)$ and the right subfigure has the $t_{t}\sim \sexpD(10/k, 0.2)$. This is a regular graph with $100$ vertices and the degree is $5$. The number of relays is $r=10$. The warden's arrival time follows $U(0,100)$. 
    It is impossible to reach maximum $P_c$ and minimum $\mathbb{E}[T_{A+B}]$ simultaneously. Sacrifice some covertness probability may bring a significant improvement to the message passing delay.}
    \label{fig:ua-com1-p1}
\end{figure*}
In Fig.~\ref{fig:ua-com1-p1}, we evaluate \eqref{Eq:model1}, \eqref{Eq:detect2} and \eqref{eq:covert1}   to see the tradeoff between the detection probability and the expected message passing delay. The same regular graph is adopted from Fig.~\ref{fig:rw-worst}. The graph on the left has a message chunk passing time $t_{t}\sim\sexpD(10/k,1)$ and the graph on the right has a message chunk passing time $t_{t}\sim \sexpD(10/k, 0.2)$. The warden's arrival time follows $U(0,30)$. We evaluate $P_c$ vs. $E[T_{A+B}]$ as $n$ increases from $1$ to $10$. Each point in the figure is a different case with a different value of $(n,k)$. Similarly, the minimum delay and maximum probability strategies are considered.

The left subfigure shows that given the number of coded chunks $n$, minimum delay and maximum probability strategies have very different optimal $k$ values. For minimum delay, the optimal $k$ is small (i.e., $1$ or $2$); For maximum probability, the optimal $k$ is equal to $n$. Therefore, we conclude that it is impossible to simultaneously increase the covertness probability and reduce the message passing delay. In practice, the appropriate $(n,k)$ is also decided by the requirements for the performance metrics, e.g., the system requires $P_c>0.7$ or $E[T_{A+B}]<170$. 
The right subfigure shows that given the number of coded chunks $n$, when $n$ is small (e.g., $n<5$), minimum delay and maximum probability strategies have very different optimal $k$ values; When $n$ is large (e.g., $n>6$), they have very closed optimal $k$ values. When $n>6$, the optimal $k$ for minimum delay is equal to $2$ and the optimal $k$ for maximum probability is equal to $5$. Therefore, we conclude that as the variance of $t_{t}$ increases ($1/\lambda^2$ becomes smaller), the maximum probability and minimum delay strategies will finally be the same. The conclusion indicates that when $\lambda$ is sufficiently small, we can simultaneously increase the covertness probability and reduce the message passing delay. 

From both subfigures, we see that the minimum delay strategy leads to similar results, but the maximum probability strategy leads to very different results. This is because $P_c$ is sensitive to the value of $\lambda$, but $E[T_{A+B}]$ is not. We analyzed how $P_c$ changes with $\lambda$ in Sec.~\ref{sub:covertprob}. According to the expression of $E[T_{A+B}]$ given in \eqref{Eq:model1}, although decreasing $\lambda$ from $1$ to $0.2$ can increase the former term $\frac{n+k}{\lambda}$, the latter term $\theta_d v(H_r+H_n-H_{r-n}-H_{n-k})$ is far larger than the former. Thus, the results will not change much with $\lambda$.
Notice that both $P_c$ and $E[T_{A+B}]$ reach the optimal values when $n$ and $k$ are small. However, that may not be the case when the message length $m$ is sufficiently large. Recall that the detection probability $P_d=1$ (in Theorem~\ref{Th:covert1}) when $\ell>W$, then the covertness probability $P_c=0$. Since $\ell=\frac{m}{k}$, when $m\gg W$ (e.g., $m=3W$), a small $(n,k)$ value (e.g., $(4,2)$) leads to $P_c=0$. Therefore, it is necessary to study how the $P_c$ and $E[T_{A+B}]$ trade-off changes when $n$ and $k$ are large in Fig.~\ref{fig:ua-com1-p1}.

\section{Conclusions and Future Work}
We introduced and studied a gossip-like protocol for covert passing messages between Alice and Bob as they move in an area hosting a multitude of Io$\beta$T objects. Alice and Bob perform random walks on random regular graphs. The Io$\beta$T objects reside on the vertices of this graph, and some can serve as relays between Alice and Bob.  In our protocol, Alice splits her message into small chunks, which she can covertly deposit to the relays she encounters. Afterward, Bob collects the chunks. Alice may encode her data before the dissemination. The area where the message passing takes place is watched over by a warden Willie. Willie can either perform random walks as Alice and Bob do or conduct uniform surveillance of the area. In either case, he can only observe one relay at a time. We evaluated the system performance by the covertness probability and the message passing delay. These performance metrics depend on the graph, communications delay, and code parameters.
We showed that, in most scenarios, it is impossible to choose the design parameters that maximize the covertness probability and minimize the message delay simultaneously.

This work sets the stage for many problems of interest to be studied in the future. We briefly describe five directions of immediate interest.

\subsubsection{Reducing Delay with Multiple Sources and Collectors}
According to Section~\ref{sub:multi}, although many random walks are faster than one \cite{alon2008many}, simply introducing multiple sources and collectors may not reduce the dissemination and collection time significantly. Fountain codes need further study as a possible way to improve the performance.

\subsubsection{Computing the Overall Message Passing Time}
As we discussed in Section~\ref{subsec:starting}, the overall time the message spends in the system depends on when the collection and dissemination start. For example, Alice and Bob can start their walks each day at some specified time, or collecting can start before the dissemination is over. Extending this work to include the dynamics on when to start each process is an exciting problem.

\subsubsection{Extending Analysis to other Detection Models and Mobility Patterns}
Some other detection models for such systems are reasonable but have not been studied yet. For example, regarding the warden, it is reasonable to assume that Willie needs to spend a certain time before being able to detect the transmission or that he can, over time, learn which nodes do not have relays.
Regarding the mobility patterns, random walks on irregular graphs or some other area traversing models are of interest.
\subsubsection{Computing trade-off Between False-alarm and Missed-detection}
Some IoT devices may help Alice and Bob achieve covert communication. For example, some devices can deceive the warden into wrongly accusing Alice and Bob of the message passing. In such scenarios, we should consider the tradeoff between false-alarm and missed-detection.
Another scenario where false-alarm and missed-detection are of interest is when the transmission is organized through incremental data redundancy in a classical way (see, e.g., 
\cite{harq:LiuSS03, harq:VarnicaSW05}). Rather than assuming a perfect detection of the message chunk, we here assume that each relay sends a very noisy version of the entire message to reduce the detection probability. The collector then recovers the message by, e.g., Chase combining of the received noisy versions.

\subsubsection{Studying Systems with Unreliable Relays }
In the Io$\beta$T scenarios, we expect some of the relays to be adversarial. In general, these objects have power constraints, and we may not rely on all of them to provide the required storage service.  Including such impairments is of interest for further study.

\label{Sec:conclusion}

\ifCLASSOPTIONcaptionsoff
  \newpage
\fi

\bibliographystyle{IEEEtran}
\bibliography{biblio.bib}

\end{document}